\documentclass[11pt]{ip-journal}
\usepackage{amsmath,amssymb}
\usepackage{graphicx, color}
\usepackage{slashed}
\newcommand{\RN}[1]{%
  \textup{\uppercase\expandafter{\romannumeral#1}}%
}

\newcommand{\R}{\mathbb{R}}
\newtheorem{theorem}{Theorem}
\newtheorem{proposition}[theorem]{Proposition}

\newtheorem{lemma}[theorem]{Lemma}

\theoremstyle{definition}
\newtheorem{definition}[theorem]{Definition}

\begin{document} 

\title[Linear Stability of Schwarzschild Spacetime]{Linear Stability of Schwarzschild Spacetime: Decay of Metric Coefficients}

\author{Pei-Ken Hung}
\address{Pei-Ken Hung\\
Department of Mathematics\\
Massachusetts Institute of Technology, USA}
\email{pkhung@mit.edu}

\author{Jordan Keller}
\address{Jordan Keller\\
Black Hole Initiative\\
Harvard University, USA}
\email{jordan\_keller@fas.harvard.edu}

\author{Mu-Tao Wang}
\address{Mu-Tao Wang\\
Department of Mathematics\\
Columbia University, USA}
\email{mtwang@math.columbia.edu}

\thanks{This material is based upon work supported by the National Science
Foundation under Grant Numbers DMS 1405152 (Mu-Tao \ Wang).  The authors would like to thank the National Center for Theoretical Sciences of National Taiwan University, where this research was initiated, for their warm hospitality.  In addition, we thank Professors Simon Brendle, Sergiu Klainerman, and Ye-Kai Wang for their interests in this work.  The third author thanks Professor Mihalis Dafermos for the reference to Johnson's work \cite{Johnson}}

\begin{abstract}
In this paper, we study the theory of linearized gravity and prove the linear stability of Schwarzschild black holes as solutions of the vacuum Einstein equations.  In particular, we prove that solutions to the linearized vacuum Einstein equations centered at a Schwarzschild metric, with suitably regular initial data, remain uniformly bounded and decay to a linearized Kerr metric on the exterior region.  We employ Hodge decomposition to split the solution into closed and co-closed portions, respectively identified with even-parity and odd-parity solutions in the physics literature.  For the co-closed portion, we extend previous results by the first two authors, deriving Regge-Wheeler type equations for two gauge-invariant master quantities without the earlier paper's need of axisymmetry.  For the closed portion, we build upon earlier work of Zerilli and Moncrief, wherein the authors derive an equation for a gauge-invariant master quantity in a spherical harmonic decomposition.  We work with gauge-invariant quantities at the level of perturbed connection coefficients, with the initial value problem formulated on Cauchy data sets.  With the choice of an appropriate gauge in each of the two portions, decay estimates on these decoupled quantities are used to establish decay of the metric coefficients of the solution, completing the proof of linear stability.  Our result differs from that of Dafermos-Holzegel-Rodnianski, both in our choice of gauge and in our identification and utilization of lower-level gauge-invariant master quantities.
\end{abstract}
\maketitle

\section{Introduction}

The Schwarzschild solution of the vacuum Einstein equation in general relativity is the unique static solution that represents an isolated gravitating system of a single black hole. Studies, both theoretically and experimentally, of such a system are modeled on the Schwarzschild solution and its perturbation. The stability of the Schwarzschild solution is thus of utmost importance. More than two decades after the work of nonlinear stability of Minkowski space by Christodoulou and Klainerman \cite{ChristKlainerman}, the nonlinear stability of Schwarzschild remains open. This paper addresses the linear stability of the Schwarzschild solution, which has a long history and rich literature involving the works of both physicists and mathematicians, culminating in the recent breakthrough of Dafermos-Holzegel-Rodnianski \cite{DHR}. This paper provides a different and simpler proof that reveals the underlying geometric structure of the vacuum Einstein equation at a more elementary level. 

The question of linear stability is formulated in the following way.  Consider the vacuum Einstein equation 
$G(g)\equiv Ric(g) -\frac{1}{2} R(g) g=0$, of which the Schwarzschild metric $g_0$ is solution. Let $\delta g$ be a solution of the linearization of the vacuum Einstein equation at Schwarzschild: \begin{equation}\label{linearized_Einstein}\delta G|_{g_0}(\delta g)=0.\end{equation} $\delta g$ is a smooth symmetric $(0, 2)$ tensor on the background Schwarzschild spacetime.  Since the Einstein equation is invariant under diffeomorphisms of the spacetime, any smooth co-vector field $X$ generates another solution $\delta g+\pi_X$ of \eqref{linearized_Einstein}, encoding an infinitesimal deformation of the spacetime via the deformation tensor $\pi_{X} = \mathcal{L}_{X}g$.  In addition, the Schwarzschild solution lies in the larger family of Kerr solutions, and there are solution of \eqref{linearized_Einstein} that correspond to Kerr perturbations. We prove linear stability in the following sense: the linearized metric coefficients of $\delta g$, after being normalized by the gauge condition, decay through a suitable foliation to a Kerr perturbation under appropriate initial conditions.

There are two main approaches to the perturbation problem:

1. Perturbation of metric coefficients: Regge-Wheeler \cite{RW} showed that in a suitable gauge, equation \eqref{linearized_Einstein} decouples into even-parity and odd-parity perturbations, corresponding to axial and polar perturbations in Chandrasekhar \cite{Chandra1}.  Originally, Regge-Wheeler \cite{RW} discovered the master Regge-Wheeler equation in the axial or odd-parity case; more than a decade later, Zerilli \cite{Zerilli} derived the eponymous equation in the polar or even-parity case.  Later work by Moncrief \cite{Moncrief} phrased the decoupling in terms of a gauge-invariant, connection-level quantity, which we refer to as the Zerilli-Moncrief function.  See also \cite{CMP, MartelPoisson, Vishveshwara}.

2. Perturbation in Newman-Penrose (N-P) formalism: In particular, it is known that the extreme linearized Weyl curvature components satisfy the Teukolsky equation \cite{Teukolsky}, which can be further solved by separation of variables. 

On the Schwarzschild background there is a transformation theory, initiated by Wald \cite{Wald2} with further refinements by Aksteiner et al. \cite{Aksteiner}, which relates the Regge-Wheeler equation and the Teukolsky equation. The Regge-Wheeler equation has a favorable potential which allows for a direct analysis, while this is not clear for the Teukolsky equation.  In this way, the mapping is vital to the work of Dafermos-Holzegel-Rodnianski \cite{DHR}, as discussed below, in addition to the work of Lousto-Whiting \cite{LoustoWhiting}.

In the Schwarzschild and Kerr settings, the estimates and techniques from the study of the scalar wave equation, regarded as a ``poor man's'' linearization of the vacuum Einstein equations, are expected to prove an essential ingredient in further progress on linear stability, with developments in linear stability playing a similar role in non-linear stability.  A complete theory, including uniform boundedness and decay estimates, is now in place for scalar waves, first appearing in the works of Dafermos-Rodnianski \cite{DR} in the Schwarzschild setting and Dafermos-Rodnianski-Shlapentokh-Rothman \cite{DRR} in sub-extremal Kerr, with contributions and refinements also appearing in \cite{KayWald, Soffer, BlueSterbenz, FSKY, AnderssonBlue, Tataru1, Tataru2, Luk}. 

To get a complete theory of linear stability, one needs to prove such decay estimates for each component of a solution $\delta g$ under a suitable gauge, and modulo the Kerr perturbations. Finster-Smoller \cite{FinsterSmoller_Sch} prove the decay estimates of Teukolsky equation on the Schwarzschild background (see also the Kerr case  \cite{FinsterSmoller}). However, the decay estimates of the perturbed metric coefficients $\delta g$ do not follow from this. Note that the extreme linearized Weyl curvature components are gauge invariant quantities, while the decay of metric coefficients holds true only after a gauge condition. In the Kerr case, it seems that the reconstruction of metric coefficients from the extreme components of Weyl curvatures in the N-P formalism remains unsolved, with partial results appearing in Wald \cite{Wald3}; see also \cite{WhitingPrice}.  Specializing to the Schwarzschild background, the results of Dafermos-Holzegel-Rodnianski \cite{DHR} in double null coordinates can be recast in the N-P formalism to yield a reconstruction procedure.

The authors of \cite{DHR} make use of the aforementioned transformation theory and show that a certain second derivative of extreme linearized Weyl curvature satisfies a Regge-Wheeler type equation with a favorable potential.  A double null gauge is imposed to derive that all perturbed metric coefficients decay modulo the Kerr perturbation. 

The current paper provides a theory of linear stability on the level of the metric perturbation. In the space of linearized solutions of symmetric $(0,2)$ tensors $\delta g$, we identify the Kerr perturbation in the subspace of lower angular modes. For the higher angular modes, we work with gauge-invariant quantities at the level of perturbed connection coefficients which satisfy Regge-Wheeler type equations with favorable potentials. The Regge-Wheeler gauge and an interpolated Chandrasekhar gauge are adopted to prove the decay of all perturbed metric coefficients and complete the proof. 

In order to identify the Kerr perturbation, we first decompose any smooth, symmetric $(0,2)$ tensor $\delta g$ according to angular modes
\begin{equation} \label{mode}
\delta g = \delta g^{\ell <2} + \delta g^{\ell \geq 2}
\end{equation} 
per Proposition \ref{mode_decomp}.

Summarizing our results, we prove linear stability of the Schwarzschild spacetime as stated in Theorem 1 and Theorem 2:

\begin{theorem}
Let $\delta g $ be a smooth, symmetric $(0,2)$ tensor on the Schwarzschild spacetime, satisfying the linearized vacuum Einstein equations \eqref{linearized_Einstein}.

For the $\delta g^{\ell < 2}$ component of $\delta g$, there exists a unique smooth co-vector $X^{\ell<2}$ (modulo Killing fields)  on the Schwarzschild spacetime and constants $c$, $d_{-1}, d_0, d_1$ such that 
\begin{equation}\label{lowerHarmonics} 
\delta g^{\ell < 2}=\pi_{X^{\ell<2}}+c K+\sum_{m=-1, 0, 1} d_m K_m,
\end{equation}
where $K, K_{-1}, K_0, K_1$ are smooth symmetric $(0, 2)$ tensors that correspond to linearized Kerr solutions specified in Definition \ref{Kerr_sol}.
\end{theorem}

The existence part of Theorem 1 is well-known, first appearing in the work of Zerilli \cite{Zerilli}.  See also Martel-Poisson \cite{MartelPoisson}.  We provide a different proof for completeness.

\begin{theorem} Under the same assumption for $\delta g$ as in Theorem 1 and assuming moreover $\delta g^{\ell \geq 2}$ is compactly supported away from the bifurcation sphere on the time-slice $\{ t = 0\}$, there exists a smooth co-vector $X^{\ell \geq 2}$  such that 
\begin{equation}
\delta g^{\ell \geq 2} = \pi_{X^{\ell \geq 2}} +\widehat{\delta g}^{\ell \geq 2},
\end{equation} 
with the components of the gauge-normalized solution $\widehat{\delta g}^{\ell \geq 2}$ decaying pointwise through a suitable foliation.
\end{theorem}
For more information on the decay mentioned above, in particular the rate of decay and the norms on initial data, we refer the reader to Section 10.

Another decomposition (the Hodge type decomposition) of the space of symmetric $(0, 2)$ tensors is adopted to study $\delta g$.  Any $\delta g$ is decomposed into the closed and co-closed portions, which generalize the even-odd or axial-polar decompositions in physics literature, without any symmetry or mode assumptions. Note that $K$ belongs to the closed part while $K_m, m=-1, 0, 1$ belong to the co-closed part.

For each portion of $\delta g^{\ell \geq 2}$, we decouple gauge-invariant quantities satisfying Regge-Wheeler type equations (\ref{RW1}, \ref{RW2}, \ref{RW3}); analysis of these Regge-Wheeler type equations shows that each quantity decays to zero through a suitable foliation.  Identification of gauge-invariant quantities of the co-closed portion appeared in the work of Regge-Wheeler \cite{RW} and Cunningham-Moncrief-Price \cite{CMP}, with estimates for the associated Regge-Wheeler equations accomplished in the mathematics literature \cite{BlueSoffer, DSS, FriedmanMorris, DHR}.  The novel feature of our work lies in exploiting these estimates for the gauge-invariant quantities identified by both Cunningham-Moncrief-Price and Regge-Wheeler to deduce estimates on a third gauge-invariant quantity via direct analysis of the linearized vacuum Einstein equations.  Identification of the gauge-invariant quantity of the closed portion appear in the works of Moncrief and Zerilli \cite{Moncrief, Zerilli}, with the present work and that of Johnson \cite{Johnson} being the first to analyze the associated Zerilli equation.  With these estimates in hand, the rest of the proof consists of the reconstruction of components of $\delta g^{\ell \geq 2}$ from these quantities under suitable gauge choice of $X^{\ell \geq 2}$ and the deduction of decay of all components of $\delta g$.  The well-known Regge-Wheeler gauge has been used in the physics literature to reconstruct the co-closed portion pointwise in terms of the gauge-invariant Regge-Wheeler and Cunningham-Moncrief-Price functions.  With the introduction and estimation of our third gauge-invariant quantity, we are further able to estimate the co-closed portion in this gauge.  The choice of  the gauge for the closed portion is more subtle.  Utilizing a novel interpolation of the Chandrasekhar gauge, as outlined in Section 9, reconstruction and estimation of the closed portion is accomplished.  In the end, the decay of all  components of  $\delta g^{\ell \geq 2}$ is achieved by imposing the Regge-Wheeler gauge on the co-closed portion and an interpolated Chandrasekhar gauge on the closed portion.

The stability of Schwarzschild spacetime is a subcase of the broader matter of Kerr stability, so it is natural to wonder about the prospects of generalizing our arguments to the Kerr background.  Owing to our reliance on the spherical symmetry of Schwarzschild spacetime in the Hodge decomposition discussed above, adaptation of our method to the Kerr setting would be nontrivial.  In the special case of small angular momentum $a << M$, with a small deviation from spherical symmetry and weak coupling of the closed and co-closed portions, the problem does appear to be tractable. 

The paper is organized as follows.  In Section 2, we present the Schwarzschild spacetimes as a family of static, spherically symmetric spacetimes satisfying the vacuum Einstein equations.  In Section 3, we discuss linearized gravity about such spherically symmetric spacetimes.  In particular, we discuss Hodge decomposition on the spheres of symmetry and decomposition into spherical harmonics.  In Section 4 we present the well-known linearized Kerr family of solutions, along with the pure gauge solutions.  Using such solutions, we treat the analysis of $\delta g^{\ell <2}$ and the proof of Theorem 1 in Section 5.  Subsequent sections deal with the analysis of the closed and co-closed portions of the remainder $\delta g^{\geq 2}$.  In Section 6, we prove decay of the co-closed portion in the Regge-Wheeler gauge, extending results from the previous \cite{HK}.  In Section 7, we present the well-known Zerilli-Moncrief function as a gauge-invariant quantity satisfying the Zerilli equation, the analysis of which is the subject of Section 8.  In Section 9, we introduce the Chandrasekhar gauge and prove decay of the closed solution under a suitable modification of the gauge.  We summarize our results on $\delta g^{\ell \geq 2}$ in Section 10, wherein we prove Theorem 2.

\section{The Schwarzschild Spacetime}

The Schwarzschild spacetimes $(\mathcal{M},g_M)$ comprise a family of static, spherically symmetric spacetimes, parametrized by mass $M > 0$.  Each such spacetime is vacuum; i.e., each metric $g_M$ satisfies the vacuum Einstein equations $Ric(g_M) = 0$.  

The staticity and spherical symmetry of the Schwarzschild family are encoded in a number of Killing fields.  In particular, we have the static Killing field, denoted $T$, and the rotational Killing fields, denoted $\Omega_{i}$, with $i = 1,2,3$.  For convenience in what follows, we collect the rotation Killing fields in the set $\Omega := \{ \Omega_{i} | i = 1,2,3\}$.  Moreover, we denote by $\mathcal{K} := \{ T, \Omega_1, \Omega_2, \Omega_3 \}$ the full set of Killing fields.

Our results concern the Schwarzschild exterior region, up to and including the future event horizon.  In the course of our analysis, various coordinate systems will prove useful; we enumerate them below.

The Schwarzschild exterior, not including the event horizon, is covered by a coordinate patch $(t,r,\theta,\phi)$ with $t\in{\R}, r > 2M, (\theta,\phi)\in{S^2}$.  In these standard Schwarzschild coordinates, the Schwarzschild metric has the form
\begin{equation}
 g_M = -\left(1-\frac{2M}{r}\right)dt^2 + \left(1-\frac{2M}{r}\right)^{-1}dr^2+ r^2\mathring\sigma_{\alpha\beta}dx^{\alpha}dx^{\beta},
 \end{equation}
where 
\begin{equation}\label{roundsphere}
\mathring{\sigma}_{\alpha\beta}dx^{\alpha}dx^{\beta} := d\theta^2 + \sin^2\theta d\phi^2
\end{equation}
is the round metric on the unit sphere.  Often we use the shorthand
\begin{align}
\mu & := \frac{2M}{r},\\
\Delta & :=  r^2 - 2Mr.
\end{align}

Alternatively, we can cover this region with the Regge-Wheeler coordinates $(t,r_{*},\theta,\phi)$, with tortoise coordinate $r_{*}$ normalized as 
\begin{equation}\label{rStar}
r_{*} := r + 2M\ln(r-2M) - 3M - 2M\ln(M),
\end{equation}
such that the metric
\begin{equation}
g_M = -(1-\mu)dt^2 + (1-\mu)dr_{*}^2 + r^2\mathring\sigma_{\alpha\beta}dx^{\alpha}dx^{\beta}
\end{equation}
is defined on $t\in{\R}, r_{*}\in{\R}, (\theta, \phi) \in {S^2}$, with $r_{*} = 0$ on the photon sphere $r = 3M$. 

A variant of the above takes $t_{*} = t + 2M\ln(r-2M)$, with 
\begin{equation}
g_M = -(1-\mu)dt_{*}^2 + 2\mu dt_{*}dr + (1+\mu)dr^2 + r^2\mathring\sigma_{\alpha\beta}dx^{\alpha}dx^{\beta},
\end{equation}
now defined for $(t_{*}, r, \theta, \phi)$ coordinates satisfying $t_{*}\in{\R}, r > 0, (\theta,\phi)\in{S^2}.$  In contrast with the previous two, this coordinate system covers both the exterior region and the black hole region.

Finally, we shall refer to the double-null coordinate system $(u,v,\theta,\phi)$, with null coordinates $u$ and $v$ related to the Regge-Wheeler coordinates by
\begin{align}
\begin{split}
u &= \frac{1}{2}(t - r_{*}),\\
v &= \frac{1}{2}(t + r_{*}).
\end{split}
\end{align}
With this relation, the Schwarzschild metric takes the form
\begin{equation}
g_M = -4(1-\mu)dudv + r^2\mathring\sigma_{\alpha\beta}dx^{\alpha}dx^{\beta}.
\end{equation}

Note that each of the above coordinate systems covers only a portion of the maximally extended Schwarzschild spacetime, globally parametrized by the Kruskal coordinates \cite{Kruskal}.  As we work only on the exterior region and the future event horizon, the coordinate systems above suffice for our purposes.  Indeed, all of the statements made below are coordinate-invariant, but are more transparently expressed in these coordinate systems rather than in the Kruskal coordinate system.  In particular, decay rates and weights appear as straightforward polynomial expressions in the above systems, while their expression in Kruskal coordinates involves transcendental functions or implicit relations.

For more information on the Schwarzschild spacetime, we direct the reader to the comprehensive references \cite{Wald, Chandra1}.

\section{Linearized Gravity in a Spherically Symmetric Background}

\subsection{Spherically symmetric background}\label{background}
The analysis in this section applies to a spherically symmetric spacetime $(\mathcal{M}, g)$ such that the group $SO(3)$ acts by isometry. 

Let $(\mathcal{Q}, \tilde{g}=\sum_{A, B=0, 1} \tilde{g}_{AB} dx^A dx^B)$ be a two-dimensional Lorentzian manifold with local coordinates $x^A, A=0,1$.  Let $(S^2, \mathring{\sigma}=\sum_{\alpha\beta= 2, 3}\mathring{\sigma}_{\alpha\beta} dx^\alpha dx^\beta)$ be the unit two-sphere with the standard Riemannian metric in local coordinates $x^\alpha, \alpha=2, 3$.  We adopt the convention that any repeated index is summed. Each point on $\mathcal{Q}$ represents an orbit sphere, with $r$ a positive function which represents the areal radius of each orbit sphere. We consider a general spherically symmetric spacetime in local coordinates $x^0, x^1, x^2, x^3$:
\begin{equation}\label{spacetime metric} 
g_{ab}dx^{a}dx^{b} = \tilde{g}_{AB} dx^A dx^B+r^2 \mathring{\sigma}_{\alpha\beta} dx^\alpha dx^\beta, A, B=0, 1, \alpha, \beta =2, 3.
\end{equation}
The index notations above are adopted throughout the paper: $A, B, C, \cdots =0, 1$ for quotient indices, $\alpha, \beta, \gamma, \cdots =2, 3$ for spherical indices, and $a, b, c, \cdots =0, 1, 2, 3$ for spacetime indices.  

The Christoffel symbols $\Gamma_{ab}^c$ of a spherically symmetric spacetime are 
\[\begin{split}\Gamma_{AB}^C & =\tilde{\Gamma}_{AB}^C,\\
\Gamma_{\alpha\beta}^\gamma&=\mathring{\Gamma}_{\alpha\beta}^\gamma,\\
\Gamma_{\alpha A}^\beta&=r^{-1} \partial_A r (\delta_\alpha^\beta),\\
\Gamma_{\alpha\beta}^D&=-r \partial^Dr (\mathring{\sigma}_{\alpha\beta}),\end{split}\]
where $\Gamma_{AB}^C$ and $\mathring{\Gamma}_{\alpha\beta}^\gamma$ are the Christoffel symbols of $\tilde{g}_{AB}$ and $\mathring{\sigma}_{\alpha\beta}$, respectively. 

In terms of these Christoffel symbols, we consider two types of differential operators,  $\tilde{\nabla}_A$ and $\mathring{\nabla}_\alpha$. When applied to functions, $\tilde{\nabla}_A$  and $\mathring{\nabla}_\alpha$ are just differentiation with respect to coordinate variables $x^A, A=0, 1$ and $x^\alpha, \alpha=2, 3$, respectively. For co-vectors, we define
\begin{equation}\label{mathring}\begin{split} \tilde{\nabla}_A dx^B&=-\tilde{\Gamma}_{AC}^B dx^C,\\
\tilde{\nabla}_A dx^\alpha&=0,\\
\mathring{\nabla}_\alpha dx^B&=0,\\
\mathring{\nabla}_\alpha dx^\beta&=-\mathring{\Gamma}_{\alpha\gamma}^\beta dx^\gamma,
\end{split}
\end{equation}
with an obvious extension of the operators to more elaborate tensor bundles.

We use the notation $\tilde\Box$ and $\mathring\Delta$ for the quotient d'Alembertian and the spherical Laplacian operators.  Furthermore, we denote the volume forms for the quotient space and the unit sphere by $\epsilon_{AB}$ and $\epsilon_{\alpha\beta}$, respectively.

Throughout the paper, we consider quantities which are scalars, co-vectors, or symmetric traceless two-tensors on the spheres of symmetry, with associated sphere bundles referred to as $\mathcal{L}(0), \mathcal{L}(-1),$ and $\mathcal{L}(-2)$, respectively.  Note that the spacetime norm $|\ |_{g}$ is positive-definite on these bundles, owing to the Riemannian nature of the orbit spheres. The bundles come equipped with projected covariant derivative operators $\slashed{\nabla}$, defined for scalars by ordinary differentiation and for co-vectors by
\[\slashed{\nabla}_a dx^\alpha =-\Gamma^{\alpha}_{a\gamma} dx^\gamma, \text{  for  } a=0, 1, 2, 3,\]
with obvious extension to symmetric traceless two-tensors.  We denote the associated d'Alembertian operators by
\begin{equation}\label{SpindAlembertian}
\slashed{\Box}_{\mathcal{L}(-s)} := \slashed{\nabla}^a\slashed{\nabla}_a,
\end{equation}
with $s = 0, 1, 2$ and the appropriate covariant derivative operator.  Note that $\slashed{\Box}_{\mathcal{L}(0)} = \Box$ is the standard d'Alembertian operator on $\mathcal{M}$.  For further details, see \cite{HK}.

The projected connection, as well as the associated d'Alembertian and Laplacian operators, are related to the quotient and spherical operators of the first subsection in a straightforward fashion.  We illustrate the procedure on the bundle $\mathcal{L}(-2)$:

\begin{align*}
\slashed{\nabla}_{A}t_{\alpha\beta} &= \partial_{A}t_{\alpha\beta} - \Gamma^{\gamma}_{A\alpha}t_{\gamma\beta} - \Gamma^{\gamma}_{A\beta}t_{\alpha\gamma}\\
&=\tilde\nabla_{A}t_{\alpha\beta} - 2r^{-1}r_{A}t_{\alpha\beta},
\end{align*}

\begin{align*}
\slashed\nabla_{B}\slashed\nabla_{A}t_{\alpha\beta} &= \partial_{B}\left(\slashed\nabla_{A}t_{\alpha\beta}\right) - \Gamma^{C}_{BA}\slashed\nabla_{C}t_{\alpha\beta}\\
&- \Gamma^{\gamma}_{BA}\slashed\nabla_{\gamma}t_{\alpha\beta} - \Gamma^{\gamma}_{B\alpha}\slashed\nabla_{A}t_{\gamma\beta} - \Gamma^{\gamma}_{B\beta}\slashed\nabla_{A}t_{\alpha\gamma}\\
&=\tilde\nabla_{B}\left(\slashed\nabla_{A}t_{\alpha\beta}\right) - 2r^{-1}r_{B}\left(\slashed\nabla_{A}t_{\alpha\beta}\right)\\
&=\tilde\nabla_{B}\tilde\nabla_{A}t_{\alpha\beta} -2r^{-1}r_{A}\tilde\nabla_{B}t_{\alpha\beta} - 2r^{-1}r_{B}\tilde\nabla_{A}t_{\alpha\beta}\\
&+ 6r^{-2}r_{A}r_{B}t_{\alpha\beta} - 2r^{-1}\left(\tilde\nabla_{A}\tilde\nabla_{B}r\right)t_{\alpha\beta},
\end{align*}

\[\slashed\nabla_{\gamma}t_{\alpha\beta} = \mathring\nabla_{\gamma}t_{\alpha\beta},\]

\begin{align*}
\slashed\nabla_{\lambda}\slashed\nabla_{\gamma}t_{\alpha\beta}&=\partial_{\lambda}\left(\slashed\nabla_{\gamma}t_{\alpha\beta}\right) - \Gamma^{\delta}_{\lambda\gamma}\slashed\nabla_{\delta}t_{\alpha\beta} \\
&- \Gamma^{\delta}_{\lambda\alpha}\slashed\nabla_{\gamma}t_{\delta\beta} - \Gamma^{\delta}_{\lambda\beta}\slashed\nabla_{\gamma}t_{\alpha\delta} - \Gamma^{A}_{\lambda\gamma}\slashed\nabla_{A}t_{\alpha\beta}\\
&= \mathring\nabla_{\lambda}\mathring\nabla_{\gamma}t_{\alpha\beta} + rr^{A}\mathring\sigma_{\lambda\gamma}\left(\slashed\nabla_{A}t_{\alpha\beta}\right)\\
&=\mathring\nabla_{\lambda}\mathring\nabla_{\gamma}t_{\alpha\beta}+rr^{A}\mathring\sigma_{\lambda\gamma}\left(\tilde\nabla_{A}t_{\alpha\beta} - 2r^{-1}r_{A}t_{\alpha\beta}\right).
\end{align*}

Contracting the above, we deduce the relation

\begin{align}\label{TwoTensorWave}
\begin{split}
\slashed{\Box}_{\mathcal{L}(-2)} t_{\alpha\beta} &= \tilde\Box t_{\alpha\beta} -2r^{-1}r^{A}\tilde\nabla_{A}t_{\alpha\beta} + r^{-2}\mathring\Delta t_{\alpha\beta} \\
&+ 2r^{-2}r^{A}r_{A}t_{\alpha\beta} - 2r^{-1}\left(\tilde\Box r\right) t_{\alpha\beta}.
\end{split}
\end{align}

Likewise, we calculate
\begin{equation}\label{CovectorWave}
\slashed{\Box}_{\mathcal{L}(-1)} v_{\alpha} = \tilde\Box v_{\alpha} + r^{-2}\mathring\Delta v_{\alpha} - r^{-1}\left(\tilde\Box r\right) v_{\alpha},
\end{equation}

\begin{equation}\label{ScalarWave}
\slashed{\Box}_{\mathcal{L}(0)} V = \Box V = \tilde\Box V + 2r^{-1}r^{A}\tilde\nabla_{A}V + r^{-2}\mathring\Delta V.
\end{equation}

Specializing to the Schwarzschild spacetime, the quotient metric in the coordinates $x^0=t, x^1=r$ has the form 
\[\tilde{g}_{AB} dx^A dx^B=-\left(1-\frac{2M}{r}\right) dt^2+\left(1-\frac{2M}{r}\right)^{-1} dr^2,\]
and the non-vanishing Christoffel symbols of the quotient metric are 
\begin{align*}
\Gamma_{11}^1&=-\frac{M}{r(r-2M)},\\
\Gamma_{00}^1&=\frac{M(r-2M)}{r^3},\\
\Gamma_{01}^0 &=\frac{M}{r(r-2M)}.
\end{align*}

For convenience in what follows, we also note the formulae
\begin{align}\label{miscellaneous}
\begin{split}
&\tilde{\nabla}_A\tilde{\nabla}_B r=\frac{M}{r^2} \tilde{g}_{AB},\\
&|\tilde{\nabla} r|^2= r^{A}r_{A} = 1-\frac{2M}{r},\\
&\tilde{\Box} r=\frac{2M}{r^2},\\
&\tilde{\Box} r^2=2,\\
&\tilde{K} = \frac{2M}{r^3}, 
\end{split}
\end{align}
with $\tilde{K}$ the Gaussian curvature of the quotient $\mathcal{Q}.$ 

\subsection{The Linearized Vacuum Einstein Equations}

Suppose a symmetric two-tensor $h_{ab}=\delta g_{ab}$ is a linear perturbation of $g_{ab}$. We recall that a perturbation of the Ricci curvature $\delta R_{bd}$ satisfies
\begin{equation}\label{linear_Ricci} 
2\delta R_{bd}= g^{ae}(\nabla_a\nabla_d h_{eb}+\nabla_a\nabla_b h_{ed}-\nabla_d \nabla_b h_{ea}-\nabla_a\nabla_e h_{bd}).
\end{equation}

The linear perturbation has the form
\begin{equation}\label{h_decomposition0}
\delta g=h_{AB} dx^A dx^B+2h_{A\alpha} dx^A dx^\alpha+h_{\alpha\beta}dx^{\alpha}dx^{\beta},
\end{equation}
with the last component admitting a further decomposition
\[h_{\alpha\beta}dx^{\alpha}dx^{\beta} = H\mathring{\sigma}_{\alpha\beta}dx^{\alpha}dx^{\beta} + \hat{h}_{\alpha\beta}dx^{\alpha}dx^{\beta}\]
into trace and traceless parts.  The trace $H$ is regarded as a function on $\mathcal{M}$, while the traceless part $\hat{h}_{\alpha\beta}$ is a symmetric traceless two-tensor with respect to $\mathring{\sigma}_{\alpha\beta}$.  Expressed in this way, the linear perturbation takes the form
\begin{equation}\label{solution}
\delta g=h_{AB} dx^A dx^B+2h_{A\alpha} dx^A dx^\alpha+H \mathring{\sigma}_{\alpha\beta} dx^\alpha dx^\beta+\hat{h}_{\alpha\beta} dx^\alpha dx^\beta,
\end{equation}
with each of $h_{AB}, h_{A\alpha}, H, \hat{h}_{\alpha\beta}$ depending upon all spacetime variables.

Perturbing about a spherically symmetric spacetime, with radial function $r$, we compute the linearized Ricci tensor:
\begin{align}
\begin{split}\label{AB}
2\delta R_{AB} &=  2r^{-1} r^{D}\left(\tilde\nabla_{A}h_{DB} + \tilde\nabla_{B} h_{DA} - \tilde\nabla_{D} h_{AB}\right) - \tilde\Box h_{AB}\\
&+ g^{CD}\left(\tilde\nabla_{C}\tilde\nabla_{A}h_{BD} + \tilde\nabla_{C}\tilde\nabla_{B}h_{AD}\right)-\tilde\nabla_{A}\tilde\nabla_{B} \left(g^{CD}h_{CD}\right)\\ 
 &- 2r^{-1} r_{A} \tilde\nabla_{B}(r^{-2}H) - 2r^{-1} r_{B} \tilde\nabla_{A}(r^{-2}H) -2\tilde\nabla_{A}\tilde\nabla_{B} \left(r^{-2}H\right)\\
 &-r^{-2}\mathring{\Delta}h_{AB} + r^{-2}\tilde\nabla_{A}\mathring\nabla^{\alpha}h_{B\alpha} + r^{-2}\tilde\nabla_{B}\mathring\nabla^{\alpha}h_{A\alpha},
\end{split}
\end{align}

\begin{align}
\begin{split}\label{Aalpha}
2\delta R_{A\alpha} = &-\tilde\nabla_{A}\mathring\nabla_{\alpha}(r^{-2}H) +\tilde\nabla^{B}\mathring\nabla_{\alpha}h_{BA} -r\tilde\nabla_{A}\mathring\nabla_{\alpha}\left(r^{-1}g^{CD}h_{CD}\right)\\
&+2r^{-1}r^{B}\tilde\nabla_{A}h_{B\alpha} + \tilde\nabla^{B}\tilde\nabla_{A}h_{B\alpha} - 2r^{-2}r_{A}r^{B}h_{B\alpha} \\
&- 2r^{-1}r_{A}\tilde\nabla^{B}h_{B\alpha} -2r^{-1}(\tilde\nabla_{A}\tilde\nabla^{B}r)h_{B\alpha} - \tilde\Box h_{A\alpha}\\
&+r^{-2}\left(\mathring\nabla^{\gamma}\mathring\nabla_{\alpha}h_{A\gamma} - \mathring\Delta h_{A\alpha}\right)
 +\tilde\nabla_{A}\mathring\nabla^{\gamma}(r^{-2}\hat{h}_{\alpha\gamma}),
\end{split}
\end{align}

\begin{align}
\begin{split}\label{alphabeta}
2\delta R_{\alpha \beta} &= \Big(2r r^{B} \tilde\nabla^{A}h_{AB} + 2r^{A}r^{B}h_{AB} + 2r(\tilde\nabla^{A}\tilde\nabla^{B}r)h_{AB}\\
&-r r^{A} \tilde\nabla_{A}(g^{CD}h_{CD}) -r^2\tilde\Box(r^{-2}H)- 2r^{-1}(\tilde\Box r)H\\
&+6r^{-2}r^{A}r_{A}H - 4r^{-1}r^{A}\tilde\nabla_{A}H-r^{-2}\mathring\Delta H \Big)\mathring\sigma_{\alpha\beta}\\
&-\mathring\nabla_{\alpha}\mathring\nabla_{\beta}(g^{CD}h_{CD}) +\left(\tilde\nabla^{A}\mathring\nabla_{\alpha}h_{A\beta} + \tilde\nabla^{A}\mathring\nabla_{\beta}h_{A\alpha}\right)\\
&+2r^{-1}r^{A}\mathring\nabla^{\gamma}h_{A\gamma}\mathring\sigma_{\alpha\beta}+2r^{-2}r^{A}r_{A}\hat{h}_{\alpha\beta} - 2r^{-1}(\tilde\Box r)\hat{h}_{\alpha\beta}\\
&-2r^{-1}r^{A}\tilde\nabla_{A}\hat{h}_{\alpha\beta} - r^{2}\tilde\Box(r^{-2}\hat{h}_{\alpha\beta})\\
&+r^{-2}\left(\mathring\nabla^{\gamma}\mathring\nabla_{\alpha}\hat{h}_{\beta\gamma} + \mathring\nabla^{\gamma}\mathring\nabla_{\beta}\hat{h}_{\alpha\gamma} - \mathring\Delta \hat{h}_{\alpha\beta}\right).
\end{split}
\end{align}

We note the commutation formulae:
\begin{align}\label{commFormulae}
\begin{split}
&\mathring\nabla^{\gamma}\mathring\nabla_{\alpha}h_{A\gamma}=\mathring\nabla_{\alpha}\mathring\nabla^{\gamma}h_{A\gamma}+h_{A\alpha},\\
&\mathring\nabla^{\gamma}\mathring\nabla_{\alpha}\hat{h}_{\beta\gamma} + \mathring\nabla^{\gamma}\mathring\nabla_{\beta}\hat{h}_{\alpha\gamma} - \mathring\Delta \hat{h}_{\alpha\beta}=\mathring\sigma_{\alpha\beta}\mathring{\nabla}^\gamma\mathring{\nabla}^\delta \hat{h}_{\gamma\delta}+2 \hat{h}_{\alpha\beta}.
\end{split}
\end{align} 

Using the second commutation formula \eqref{commFormulae}, we split the equation \eqref{alphabeta} into trace and traceless parts:
\begin{align}\label{alphabetaTrace}
\begin{split}
\delta R_{\alpha\beta}\mathring\sigma^{\alpha\beta} &= 2r r^{B} \tilde\nabla^{A}h_{AB} + 2r^{A}r^{B}h_{AB} + 2r(\tilde\nabla^{A}\tilde\nabla^{B}r)h_{AB}\\
&-r r^{A} \tilde\nabla_{A}(g^{CD}h_{CD}) -r^2\tilde\Box(r^{-2}H)- 2r^{-1}(\tilde\Box r)H\\
&+6r^{-2}r^{A}r_{A}H - 4r^{-1}r^{A}\tilde\nabla_{A}H-r^{-2}\mathring\Delta H+ \tilde{\nabla}^{A}\mathring\nabla^{\alpha}h_{A\alpha}\\
 &+2r^{-1}r^{A}\mathring\nabla^{\gamma}h_{A\gamma} + r^{-2}\mathring\nabla^{\gamma}\mathring\nabla^{\delta}\hat{h}_{\gamma\delta} -\frac{1}{2}\mathring\Delta(g^{CD}h_{CD}),
 \end{split}
 \end{align}
 \begin{align}\label{alphabetaTraceless}
 \begin{split}
 2\widehat{\delta R}_{\alpha\beta} &= 2\delta R_{\alpha\beta} - \delta R_{\gamma\delta}\mathring\sigma^{\gamma\delta}\mathring\sigma_{\alpha\beta}\\
 &=  \tilde{\nabla}^{A}\left(\mathring\nabla_{\alpha}h_{A\beta} + \mathring\nabla_{\beta} h_{A\alpha} - \mathring\nabla^{\gamma}h_{A\gamma}\mathring\sigma_{\alpha\beta}\right)+ 2r^{-2}\hat{h}_{\alpha\beta}\\
 &+2r^{-2}r^{A}r_{A}\hat{h}_{\alpha\beta} - 2r^{-1}(\tilde\Box r)\hat{h}_{\alpha\beta} -2r^{-1}r^{A}\tilde\nabla_{A}\hat{h}_{\alpha\beta}\\
 &-\mathring\nabla_{\alpha}\mathring\nabla_{\beta}(g^{CD}h_{CD}) + \frac{1}{2}\mathring\Delta(g^{CD}h_{CD})\mathring\sigma_{\alpha\beta} - r^{2}\tilde\Box(r^{-2}\hat{h}_{\alpha\beta}).
 \end{split}
 \end{align}

\subsection{Hodge Decomposition}

Recall that a general symmetric two-tensor on a spherically symmetric spacetime has the pointwise decomposition 
 \begin{equation}\label{h_decomposition} \delta g=h_{AB} dx^A dx^B+2h_{A\alpha} dx^A dx^\alpha+(H\mathring\sigma_{\alpha\beta}+ \hat{h}_{\alpha\beta}) dx^\alpha dx^\beta,\end{equation}
 with each linearized metric coefficient depending on all spacetime variables.  
 
The following Hodge type decomposition of components of $\delta g$ having the form \eqref{h_decomposition} is derived in the appendix:

\begin{proposition}\label{Hodge}
Let $h_{A\alpha} dx^A dx^\alpha$ be a two-tensor on $\mathcal{M}$.  Regarding $h_{A\alpha}dx^{\alpha}$ as a co-vector on $S^2$ for each $A = 0,1$, there exist functions $H_A$ and $\underline{H}_A$ on $\mathcal{M}$ such that
\begin{equation}\label{crossDecomposition}
h_{A\alpha} dx^{A}dx^\alpha=[\mathring{\nabla}_\alpha H_A+\epsilon_\alpha^\beta (\mathring{\nabla}_\beta \underline{H}_A)] dx^{A}dx^\alpha.
\end{equation}

Let $\hat{h}_{\alpha\beta} dx^\alpha dx^\beta$ be a two-tensor on $\mathcal{M}$.  Regarding $\hat{h}_{\alpha\beta}dx^{\alpha}dx^{\beta}$ as a symmetric traceless two-tensor on $S^2$, in the sense that $\mathring{\sigma}^{\alpha\beta}\hat{h}_{\alpha\beta}=0$, there exist functions $H_2$, $\underline{H}_2$ on $\mathcal{M}$ such that 
\begin{align}\label{tracelessDecomposition}
\begin{split}
\hat{h}_{\alpha\beta} dx^\alpha dx^\beta &= \Big[  (\mathring{\nabla}_{\alpha}\mathring{\nabla}_\beta {H}_2-\frac{1}{2}\mathring{\sigma}_{\alpha\beta} \mathring{\Delta} {H}_2)\\
&+\frac{1}{2}  (\mathring{\nabla}_{\alpha}\epsilon_\beta^\gamma \mathring{\nabla}_\gamma \underline {H}_2+  \mathring{\nabla}_{\beta}\epsilon_\alpha^\gamma \mathring{\nabla}_\gamma \underline{H}_2)\Big] dx^\alpha dx^\beta. 
\end{split}
\end{align}
\end{proposition}

\begin{proposition}
Any symmetric two-tensor $\delta g$ of the form \eqref{h_decomposition} can be decomposed as $\delta g=h_1+h_2$ where 
\begin{align}\label{h_1}
\begin{split}
h_1 &=h_{AB} dx^A dx^B+2 (\mathring{\nabla}_\alpha H_A)dx^\alpha dx^A\\
&+(H\mathring{\sigma}_{\alpha\beta}+\mathring{\nabla}_{\alpha}\mathring{\nabla}_\beta {H}_2-\frac{1}{2}\mathring{\sigma}_{\alpha\beta} \mathring{\Delta} {H}_2) dx^\alpha dx^\beta,
\end{split}
\end{align}
\begin{equation}\label{h_2} h_2=2  \epsilon_\alpha^\beta (\mathring{\nabla}_\beta \underline{H}_A)dx^\alpha dx^A+\frac{1}{2}  (\mathring{\nabla}_{\alpha}\epsilon_\beta^\gamma \mathring{\nabla}_\gamma \underline {H}_2+  \mathring{\nabla}_{\beta}\epsilon_\alpha^\gamma \mathring{\nabla}_\gamma \underline{H}_2) dx^\alpha dx^\beta.\end{equation}
\end{proposition}

We note that the total number of components of $\delta g$ remains ten, given by $h_{AB}, A, B=0, 1$, $H_A, A=0, 1$, $\underline{H}_A, A=0, 1$, $H$, $H_2$, and $\underline{H}_2$. 

We refer to the portions $h_1$ and $h_2$ into which $\delta g$ was decomposed above as the closed and co-closed portions, respectively.  As the decomposition is invariant under the spacetime covariant derivative, the closed and co-closed portions are themselves solutions of the linearized vacuum Einstein equations \eqref{linearized_Einstein}.  Working within a linear theory, there is no trouble in studying each of these pieces separately; we simply add the two together to recover the original. 

The closed and co-closed solutions generalize the even-parity (polar) and the odd-parity (axial) solutions in the physics literature, respectively.  In the course of the paper, we will simply refer to them as closed and co-closed solutions.

\subsection{Spherical Harmonics}
In this subsection we recall the tensor spherical harmonics, first introduced by Regge-Wheeler \cite{RW}.  See also \cite{MartelPoisson}.

The scalar spherical harmonics $Y^{\ell m}$, indexed by integers $\ell \geq 0$ and $|m|\leq \ell$, are eigenfunctions of the spherical Laplacian, with eigenvalue $-\ell(\ell + 1)$.  That is, the $Y^{\ell m}$ satisfy
\begin{equation}
\mathring{\Delta} Y^{\ell m} = -\ell(\ell+1)Y^{\ell m}
\end{equation}
for $\ell \geq 0, |m|\leq \ell$.

With the normalization
\begin{equation}
||Y^{\ell m}||_{L^2(S^2)} = 1,
\end{equation}
the eigenfunctions $Y^{\ell m}$ form a complete, orthonormal basis of $L^2(S^2)$.

For co-vectors on the sphere, we have the closed harmonics
\begin{equation}
Y_{\alpha}^{\ell m} := \mathring\nabla_{\alpha}Y^{\ell m},
\end{equation}
and the co-closed harmonics
\begin{equation}
X_{\alpha}^{\ell m} := \epsilon_{\alpha\beta}\mathring\nabla^{\beta}Y^{\ell m}.
\end{equation}

Regarding symmetric traceless two-tensors on the sphere, closed harmonics have the form
\begin{equation}
Y_{\alpha\beta}^{\ell m} := \left(\mathring\nabla_{\alpha}\mathring\nabla_{\beta} +\frac{1}{2}\ell(\ell +1)\mathring\sigma_{\alpha\beta}\right)Y^{\ell m},
\end{equation}
while, for the co-closed harmonics,
\begin{equation}
X_{\alpha\beta}^{\ell m} := \left(\epsilon_{\alpha}^{\gamma}\mathring\nabla_{\beta} + \epsilon_{\beta}^{\gamma}\mathring\nabla_{\alpha}\right)\mathring\nabla_{\gamma} Y^{\ell m}.
\end{equation}

We note that, for the co-vector harmonics
\begin{align}
\begin{split}
&\mathring\Delta Y^{\ell m}_{\alpha} = (1-\ell(\ell + 1))Y^{\ell m}_{\alpha},\\
&\mathring\Delta X^{\ell m}_{\alpha} = (1-\ell(\ell + 1))X^{\ell m}_{\alpha},
\end{split}
\end{align}
with support on $\ell \geq 1$.  For the symmetric traceless two-tensor harmonics, we have
\begin{align}
\begin{split}\label{twotensorLaplace}
&\mathring\Delta Y^{\ell m}_{\alpha\beta} = (4-\ell(\ell + 1))Y^{\ell m}_{\alpha\beta},\\
&\mathring\Delta X^{\ell m}_{\alpha\beta} = (4-\ell(\ell + 1))X^{\ell m}_{\alpha\beta},
\end{split}
\end{align}
with support on $\ell \geq 2$.  For more details, see \eqref{oneformHarmonics} and \eqref{twoformHarmonics} in the appendix.

Using the spherical harmonic decomposition above, we split the linearized metric as
\begin{equation}
\delta g = \delta g^{\ell < 2} + \delta g^{\ell \geq 2},
\end{equation}
according to the following proposition.

\begin{proposition}\label{mode_decomp} Any symmetric two-tensor $\delta g$ on a spherically symmetric spacetime can be decomposed into $\delta g = \delta g^{\ell <2} + \delta g^{\ell \geq 2}$, in which the components of \[ \delta g^{\ell \geq 2}=h_{AB} dx^A dx^B+2h_{A\alpha} dx^A dx^\alpha+h_{\alpha\beta} dx^\alpha dx^\beta\] are characterized by the vanishing of the integrals
\begin{align*}
\int_{S^2} h_{AB} Y^{\ell m} &= 0, \\
\int_{S^2}  (\mathring{\nabla}^ \alpha h_{A\alpha}) Y^{\ell m} &= 0,\\
\int_{S^2}  (\epsilon^{\alpha\beta} \mathring{\nabla}_\alpha h_{A\beta}) Y^{\ell m} &= 0, \\
\int_{S^2} \mathring{\sigma}^{\alpha\beta} h_{\alpha\beta} Y^{\ell m} &= 0,
\end{align*}
with respect to the scalar spherical harmonics $Y^{\ell m}$ with $\ell < 2$.
\end{proposition}

\section{Linearized Kerr Solutions and Pure Gauge Solutions}

\subsection{Linearized Kerr Solutions}

Considering the Boyer-Lindquist coordinates as an extension of the standard Schwarzschild coordinates, we write the Kerr metric in a form suggestive of \eqref{h_decomposition}:
\begin{align}
\begin{split}
g_{M,a} &= -\left(1-\frac{2M}{r}\right)dt^2 + \left(1-\frac{2M}{r}\right)^{-1}dr^2 + r^2\mathring\sigma_{\alpha\beta}dx^{\alpha}dx^{\beta}\\
&- \frac{4Ma}{r}\sin^2\theta d\phi dt + O(a^2).
\end{split}
\end{align}

We treat separately the linearized change in mass and change in angular velocity below.

\subsubsection{Linearized Change in Mass}
In the expression above, linearized mass solutions have the form
\begin{equation}
h_{ab}dx^adx^b = \frac{\delta M}{r} dt^2+\frac{\delta M r}{(r-2M)^2} dr^2,
\end{equation}
giving infinitesimal change in mass within the Schwarzschild family.  We can verify directly that the symmetric two-tensor $h_{ab}$ satisfies the linearized vacuum Einstein equations \eqref{linearized_Einstein}.  Note that linearized Schwarzschild solutions are closed solutions, supported at the lowest harmonic $\ell = 0$.

\subsubsection{Linearized Change in Angular Velocity}
Infinitesimal change in angular velocity appears in the linearized Kerr solution as
\begin{align}
\begin{split}
h_{ab}dx^a dx^b &= \frac{\delta a}{r}\epsilon_{\alpha}^{\beta}Y_{\beta}^{1m}dx^{\alpha}dt\\
&=\frac{\delta a}{r}X^{1m}_{\alpha}dx^{\alpha}dt,
\end{split}
\end{align}
for $m = -1, 0, 1$.  Again, direct computation shows that $h_{ab}$ is a solution of the linearized vacuum Einstein equations \eqref{linearized_Einstein}.  Such linearized Kerr solutions are co-closed solutions, supported at the harmonic $\ell = 1$.

The above linear perturbations of the Schwarzschild metric form the four dimensional family of linearized Kerr solutions. 
\begin{definition}\label{Kerr_sol} The linearized Kerr solutions $K, K_{-1}, K_0, K_1$ of the linearized vacuum Einstein equations on Schwarzschild are given by:
\begin{align*}
&K=\frac{1}{r} dt^2+\frac{r}{(r-2M)^2} dr^2, \\
&K_m=\frac{1}{r} \epsilon_\alpha^\beta \mathring{\nabla}_\beta Y^{1m} dx^\alpha dt,
\end{align*}
with $m = -1, 0 ,1.$
\end{definition}

\subsection{Pure Gauge Solutions}
The following calculation lemma characterizes pure gauge solutions:
\begin{lemma} \label{lemmaG}
Suppose $G$ is a co-vector on Schwarzschild, with
\begin{equation}\label{G}
G =G_A dx^A+(\mathring{\nabla}_\alpha G_2) dx^\alpha + (\epsilon_\alpha^\beta\mathring{\nabla}_\beta \underline{G}_2) dx^\alpha.
\end{equation}

Then the deformation tensor of $G$ decomposes into closed and co-closed parts, 
$\pi_G=\pi_1+\pi_2$, of the form
\begin{equation}\label{pi_1}
\begin{split}
\pi_1&= (\tilde{\nabla}_A G_B+\tilde{\nabla}_B G_A) dx^A dx^B\\
&+ \mathring{\nabla}_\alpha [\tilde{\nabla}_A G_2-2 (r^{-1} \partial_A r) G_2+G_A] dx^A dx^\alpha\\
&+ 2[\mathring{\nabla}_\alpha\mathring{\nabla}_\beta G_2+r(\partial^A r) G_A \mathring{\sigma}_{\alpha\beta}]dx^\alpha dx^\beta, 
\end{split}
\end{equation}
\begin{equation}\label{pi_2}
\begin{split}\pi_2&= \epsilon_\alpha^\beta \mathring{\nabla}_\beta [\tilde{\nabla}_A \underline{G}_2 -2 (r^{-1} \partial_A r) \underline{G}_2] dx^A dx^\alpha\\
&+ [  \mathring{\nabla}_{\alpha}\epsilon_\beta^\gamma \mathring{\nabla}_\gamma \underline {G}_2+  \mathring{\nabla}_{\beta}\epsilon_\alpha^\gamma \mathring{\nabla}_\gamma \underline{G}_2]dx^\alpha dx^\beta.
\end{split}
\end{equation}
\end{lemma}

To ensure decay, we must exhaust all gauge freedom by fixing a gauge for both the closed and co-closed portions above.  We treat such gauge fixing, as well as an identification of the linearized Kerr parameters, for $\delta g^{\ell < 2}$ in the next section, deferring gauge fixing of $\delta g^{\ell \geq 2}$ until later in the paper.

\section{Analysis of the Lower Harmonics and Proof of Theorem 1}\label{lowerModesSection}

In this section, we analyze the lower harmonics represented in $\delta g^{\ell <2}$.  With the addition of a suitable pure gauge solution, we are able to extract the associated linearized Kerr parameters in Definition \ref{Kerr_sol}, encoding linearized change in mass and angular velocity, and prove Theorem 1. 

\subsection{The $\ell = 0$ Case}
\begin{proposition}
Let $\delta g$ be a smooth solution of the linearized vacuum Einstein equations \eqref{linearized_Einstein} on Schwarzschild, supported at $\ell=0$.  Then $\delta g$ is decomposable into a linearized Kerr solution and a pure gauge solution; that is,
\[\delta g=\pi_X+c K\] for a constant $c$,  a smooth co-vector field $X$, and the linearized Kerr solution $K$ in Definition \ref{Kerr_sol}.
\end{proposition}
\begin{proof}
With the hypotheses above, $\delta g$ has the form \[\delta g=h_{AB} dx^A dx^B+H\mathring\sigma_{\alpha\beta} dx^\alpha dx^\beta,\] where $h_{AB}$ and $H$ are supported at $\ell=0$.  We eliminate the trace $H$ and diagonalize $h_{AB}$ by the co-vector field $X'=G_Adx^A$, with 
\begin{align*}
&G_0=\left(1-\frac{2M}{r}\right) \int_{3M}^r \frac{1}{1-\frac{2M}{s}}[ h_{01}(t, s)-\frac{1}{2(s-2M)}\partial_t H (t, s)]ds,\\
&G_1=\frac{1}{2(r-2M)} H.
\end{align*}
Direct calculation shows that 
\begin{equation} \label{reduction_0} 
\delta g-\pi_{X'} =h^*_{00} dt^2+h^*_{11} dr^2.
\end{equation}
Note that there is residual gauge freedom, as co-vectors $\bar{X}$ of the form 
\begin{align}\label{gaugefreedom0}
\bar{X}=\left(1-\frac{2M}{r}\right)\bar{c}(t)dt,
\end{align}  
with $\bar{c}(t)$ an arbitrary function of time, preserve this gauge reduction.

Projecting the linearized Einstein equations to the harmonic $\ell = 0$ yields
\begin{align}
\begin{split}
&\delta R_{01}: \partial_{t}h^*_{11} = 0,\\
&\delta G_{11}: \partial_{r}\left(\frac{r}{r-2M}h^*_{00}\right) + \frac{1}{r}h^*_{11} = 0,\\
&\delta R_{\alpha\beta}: \partial_{r}\left(\frac{r}{r-2M}h^*_{00}\right) + \partial_{r}\left(\frac{r-2M}{r}h^*_{11}\right) + \frac{2}{r}h^*_{11} = 0.
\end{split}
\end{align}
The system has the general solution
\begin{align}
\begin{split}
&h^*_{11} = \frac{cr}{(r-2M)^2},\\
&h^*_{00}=\frac{c}{r}+\left(1-\frac{2M}{r}\right) c(t),
\end{split}
\end{align}
for any constant $c$ and an arbitrary function $c(t)$ of $t$. That is, \[\delta g-\pi_{X'}=cK +\left(1-\frac{2M}{r}\right) c(t) dt^2.\]

For a co-vector $\bar{X}$ of the form $\left(1-\frac{2M}{r}\right) \bar{c}(t)dt$, the deformation tensor $\pi_{\bar{X}}$ is $2(1-\frac{2M}{r}) \bar{c}'(t) dt^2$. Choosing $\bar{c}(t)$ such that $\bar{c}'(t)=\frac{c(t)}{2}$, and letting $X = X' + \bar{X}$, we conclude that 
\[\delta g-\pi_{X}=c K.\]
Note that we still have gauge freedom, in the form $\bar{c}(t) \equiv \bar{c}$.  Such transformations correspond to scalar multiples of the static Killing field $T$, with vanishing deformation tensor.  Modulo these translations, the decomposition of $\delta g$ above is unique.
\end{proof} 

\subsection{The $\ell = 1$, Closed Case}
\begin{proposition}
Let $\delta g$ be a smooth, closed solution of the linearized vacuum Einstein equations \eqref{linearized_Einstein} on Schwarzschild, supported at $\ell=1$.  Then $\delta g$ is a pure gauge solution; that is,
\[\delta g=\pi_X\] for a smooth co-vector field $X$.
\end{proposition}

\begin{proof}
    We reduce the solution using the Chandrasekhar gauge outlined in the Section \ref{Chandra_gauge}.  Briefly, we are able to choose a co-vector $X'$ of the form $G_A dx^A+(\mathring{\nabla}_\alpha G_2) dx^\alpha $ such that
\[\delta g-\pi_{X'}=     h^{*}_{00} dt^2+h^{*}_{11} dr^2+H^{*}\mathring\sigma_{\alpha\beta} dx^\alpha dx^\beta\] along with the initial value condition \begin{equation}\label{initial} H^{*}+(r^2-2Mr)h^{*}_{11}=0\end{equation} on the time-slice $\{ t = 0\}$; see Lemma \ref{ChandraGauge} for details.  In particular, $G_2$ is a solution of  the inhomogeneous equations \eqref{hyper_eq} and \eqref{initial_eq}, and $G_0$ and $G_1$ are given by \eqref{G_A}. The residual gauge, or the solution of the corresponding homogeneous equations, is of the form \eqref{residual} in which
\begin{equation}\label{closedgaugefreedom1}
\bar{X}=-r^2\tilde{\nabla}_A(r^{-2} G )dx^A+(\mathring{\nabla}_\alpha  G)  dx^\alpha,\\
\end{equation} with \[G =r^{1/2}\left(r-2M\right)^{1/2}     \left[ \bar{c}_1(\theta,\phi)   p(r)+\bar{c}_2(t,\theta,\phi)\right],\]
where 
\[p(r)=\int_{3M}^r s^{1/2}\left(s-2M\right)^{-3/2} ds\]
and $\bar{c}_1(\theta,\phi)$ and $\bar{c}_2(t,\theta,\phi)$ are supported on $\ell=1$. We compute
\[\pi_{\bar{X}}= \pi_{00}  dt^2-\frac{1}{r(r-2M)}  H^{\bar{X}}  dr^2+H^{\bar{X}} \mathring\sigma_{\alpha\beta} dx^\alpha dx^\beta,\]where \begin{equation}\label{h_X} H^{\bar{X}}=-6M r^{-1/2}(r-2M)^{1/2}[ \bar{c}_1 p(r)+\bar{c}_2]-2\bar{c}_1r.\end{equation}
and 
\begin{align}
\begin{split}\label{pi_00}
\pi_{00}&=-2  r^{1/2} (r-2M)^{1/2} (\partial_t^2 \bar{c}_2)\\
&+2 M(-1+3Mr^{-1}) r^{-5/2}(r-2M)^{1/2}[ \bar{c}_1 p(r)+\bar{c}_2]+2 Mr^{-2} \bar{c}_1.
\end{split}
\end{align}

 Under this gauge condition, the $\delta R_{0\alpha}$ component of the linearized Ricci tensor
gives 
\[\partial_t (H^{*}+(r^2-2Mr)h^{*}_{11})=0.\]

On the other hand,  the $\delta R_{01}$, $\delta G_{11}$, $\delta R_{1\alpha}$ and $\delta R_{00}$  components give the following two equations for $H^*$:
\begin{align}\label{first_set}
\begin{split}
&\partial_t \left[ \partial_{r}(r^{-2}H^{*}) + \frac{2r-5M}{r^3(r-2M)}H^{*} \right]=0,\\
&\partial_r\left((r^2-2Mr)  \left[ \partial_r( r^{-2}H^{*}) + \frac{2r-5M}{r^3(r-2M)}H^{*}\right]\right)=0.
\end{split}
\end{align}

See Proposition \ref{gaugefreedom} for details. The system implies 
\[(r^2-2Mr)\left[ \partial_r( r^{-2}H^{*}) + \frac{2r-5M}{r^3(r-2M)}H^{*} \right]=c_1(\theta, \phi)\]
for $c_1$ independent of $t, r$ and supported at $\ell=1$, with general solution $H^{*}$ of the form \eqref{h_X}. 
The $\delta R_{1\alpha}$ and $\delta G_{11}$ components give the following equations:

\begin{equation}
\begin{split}
&\partial_r(\frac{r}{r-2M} h^*_{00})+ \frac{M}{r^2} h^*_{11} +\frac{M}{(r-2M)^2} h^*_{00}\\
& +r^{-1}(-\frac{r}{r-2M} h^*_{00}+\frac{r-2M}{r} h^*_{11})   - \partial_r (r^{-2}H^*)=0\\
&-\frac{2}{r^2}h^*_{11}+\frac{2}{(r-2M)^2}h^*_{00}-\frac{2}{r}\partial_r\left( \frac{r}{r-2M}h^*_{00} \right)\\
&+\frac{2(r-M)}{r^2-2Mr}\partial_r(r^{-2}H^*)-\frac{2}{(r-2M)^2}\partial_t^2 H^*=0
\end{split}
\end{equation}


Replacing $h^*_{11}$ with $-(r^2-2Mr)^{-1} H^*$ and solving $h^*_{00}$ in terms of $H^*$ shows that $h^*_{00}$ is of the form \eqref{pi_00}.  We are thus able to account for $\delta g-\pi_{X'}$ by exercising our residual gauge freedom.  That is, with appropriate choices of $\bar{c}_1$ and $\bar{c}_2$ in $\bar{X}$, we define $X = X' + \bar{X}$ such that $\delta g = \pi_{X}$.
\end{proof}

\subsection{The $\ell = 1$, Co-closed Case}
\begin{proposition}
Let $\delta g$ be a smooth, co-closed solution of the linearized vacuum Einstein equations \eqref{linearized_Einstein} on Schwarzschild, supported at $\ell=1$.  Then $\delta g$ is decomposable into a linearized Kerr solution and a pure gauge solution; that is,
\[\delta g=\pi_X+\sum_{m = -1,0,1}d_{m} K_{m}\] for constants $d_{m}$,  a smooth co-vector field $X$, and the linearized Kerr solutions $K_{m}$ in Definition \ref{Kerr_sol}.
\end{proposition}

\begin{proof}
Projecting to the harmonic $\ell = 1$, the co-closed portion of $\delta g$ has the form
\[\delta g = 2 (\epsilon_\alpha^\beta \mathring{\nabla}_\beta \underline{H}_A) dx^\alpha dx^A,\] where $\underline{H}_A$ are supported at $\ell=1$.  We reduce $\delta g$ using the co-vector field $X' = \epsilon_{\alpha}^{\beta}\mathring\nabla_{\beta} Gdx^{\alpha}$, with 
\[\partial_r G-\frac{2}{r}  G=\underline{H}_1,\] such that
\[\delta g-\pi_{X'} =2 (\epsilon_\alpha^\beta \mathring{\nabla}_\beta \underline{H}^*_0) dx^\alpha dt.\]

Co-vectors of the form
\begin{align}
\bar{X} =   \epsilon_{\alpha}^{\beta}\mathring\nabla_{\beta} ( r^2 \sum_m \bar{c}_m (t) Y^{1m}) dx^{\alpha},
\end{align} with $\bar{c}_{m}(t)$ an arbitrary function of time, act as residual gauge solutions, preserving the elimination of $\underline{H}_1$ above.

Writing $\underline{H}^{*}_A=\sum_m \underline{H}_A^{*m} (t, r) Y^{1m}$, the linearized vacuum Einstein equations (see Section 6.1) amount to
\[\epsilon^{CD} \tilde\nabla^{B}\left[ r^4 \tilde{\nabla}_D(r^{-2} \underline{H}^{*m}_{C})\right]=0, B=0, 1.\] Therefore, 
\[\epsilon^{CD} r^4 \tilde{\nabla}_D(r^{-2} \underline{H}^{*m}_{C})=d_{m},\] for a constant $d_m$. With $\underline{H}^{*}_1=0$, we deduce
\[\partial_r (r^{-2} \underline{H}^{*m}_0)={d_{m}}r^{-4},\]
with general solution
\[\underline{H}^{*m}_{0} = \frac{d_{m}}{r} + c_m(t) r^2.\]  That is,
\[\delta g-\pi_{X'}= \sum_{m = -1,0,1} \left(d_{m}K_{m} + c_m(t) r^2\right).\]

Taking $\bar{X}$ with $\bar{c}_m'(t)=c_m(t)$, and letting $X = X' + \bar{X}$, we find
\[\delta g-\pi_{X}= \sum_{m = -1,0,1} d_{m}K_{m}.\]
Again, there remains gauge freedom in the form $\bar{c}(t) \equiv \bar{c}$.  Such transformations correspond to scalar multiples of the angular Killing fields $\Omega_{i}$, with vanishing deformation tensor.  Modulo these rotations, the decomposition of $\delta g$ in the proposition is unique.
\end{proof}

\subsection{Proof of Theorem 1}

Combining the propositions above, we have a proof of Theorem 1.  In particular, adding the various linearized solutions, we find that there exists a smooth co-vector $X^{\ell<2}$ (unique modulo Killing fields) on the Schwarzschild spacetime and constants $c$, $d_{-1}, d_0, d_1$ such that 
\begin{equation}\label{lowerHarmonics} 
\delta g^{\ell < 2}=\pi_{X^{\ell<2}}+c K+\sum_{m=-1, 0, 1} d_m K_m,
\end{equation}
where $K, K_{-1}, K_0, K_1$ are smooth symmetric two-tensors that correspond to linearized Kerr solutions specified in Definition \ref{Kerr_sol}.

The rest of the paper is concerned with the analysis of the closed and co-closed pieces of the remainder $\delta g^{\ell \geq 2}$.

\section{The Co-closed Solution}

In this section, we analyze the co-closed portion $h_2$ \eqref{h_2} of $\delta g^{\ell \geq 2}$.  Recall that the co-closed portion has vanishing components
\begin{align}
\begin{split}
&h_{AB} = 0,\\
&H= 0,
\end{split}
\end{align}
with its remaining components satisfying the divergence conditions
\begin{align}
\begin{split}
&\mathring\nabla^{\alpha}h_{A\alpha} = 0,\\
&\mathring\nabla^{\alpha}\mathring\nabla^{\beta}\hat{h}_{\alpha\beta} = 0.
\end{split}
\end{align}

\subsection{The Linearized Vacuum Einstein Equations} 

With the vanishing of $h_{AB}$ and $H$ and the divergence-free conditions above, we need only consider the equations \eqref{Aalpha} and \eqref{alphabetaTraceless}.  After some simplification, we find
\begin{align}
\begin{split}\label{Aalpha_2}
2\delta R_{A\alpha} = &-r^{-2}\epsilon_{AB}\epsilon^{CD} \tilde\nabla^{B}\left( r^4 \tilde{\nabla}_D(r^{-2} h_{C\alpha})\right) \\
&-r^{-2}\mathring\nabla^{\beta}\left(\mathring\nabla_{\alpha}h_{A\beta}+\mathring\nabla_{\beta} h_{A\alpha} - r^{2}\tilde\nabla_{A}(r^{-2}\hat{h}_{\alpha\beta})\right)\\
&- r^{-2}\tilde\Box (r^2)h_{A\alpha} + 2r^{-2}h_{A\alpha},
\end{split}
\end{align}
\begin{align}
\begin{split}\label{alphabeta_2}
2\widehat{\delta R}_{\alpha\beta} &= \tilde{\nabla}^{A}\left(\mathring\nabla_{\alpha}h_{A\beta} +  \mathring\nabla_{\beta}h_{A\alpha} - r^2\mathring\nabla_{A}(r^{-2}\hat{h}_{\alpha\beta})\right)\\ 
&- 2r^{-2}r^{A}r_{A}\hat{h}_{\alpha\beta} -2r^{-1}(\tilde\Box r)\hat{h}_{\alpha\beta} + 2r^{-2}\hat{h}_{\alpha\beta}.
\end{split}
\end{align}

In \eqref{Aalpha_2} we have used the generic calculation
\[\begin{split}&r^{-2}\epsilon_{AB}\epsilon^{CD} \tilde\nabla^{B}\left( r^4 \tilde{\nabla}_D(r^{-2} h_{C\alpha})\right) = \tilde{\Box} h_{A\alpha}-\tilde{\nabla}^B\tilde{\nabla}_A h_{B\alpha}
-r^{-2}\tilde\Box(r^2)h_{A\alpha}\\
 &+ 2r^{-2}r^{B}r_{A}h_{B\alpha} + 2r^{-1}(\tilde\nabla^{B}\tilde\nabla_{A}r)h_{B\alpha}+ 2r^{-1}r_{A}\tilde\nabla^{B}h_{B\alpha} - 2r^{-1}r^{B}\tilde\nabla_{A}h_{B\alpha},\end{split}\]
in addition to the first commutation formula \eqref{commFormulae}.

Next, we define the spherical one-form
\begin{equation}\label{Pdef}
P = P_{\alpha}dx^{\alpha} := r^3\epsilon^{AB}\tilde\nabla_{B}(r^{-2}h_{A\alpha})dx^{\alpha},
\end{equation}
and the mixed quantity
\begin{align}\label{Qdef}
\begin{split}
Q &= Q_{\alpha\beta A}dx^{\alpha}dx^{\beta}dx^{A} \\
&:= \left(\mathring\nabla_{\beta} h_{A\alpha} + \mathring\nabla_{\alpha} h_{A\beta} - r^{2}\tilde\nabla_{A}(r^{-2}\hat{h}_{\alpha\beta})\right)dx^{\alpha}dx^{\beta}dx^{A},
\end{split}
\end{align}
each of which is gauge-invariant with respect to co-closed pure gauge solutions \eqref{pi_2}.

Specializing to the Schwarzschild background (that is, applying \eqref{miscellaneous}), the linearized vacuum Einstein equations \eqref{linearized_Einstein} take the form
\begin{equation}\label{one}
2\delta R_{A\alpha} = -r^{-2}\epsilon_{AB}\tilde\nabla^{B}(r P_{\alpha}) - r^{-2}\mathring\nabla^{\beta}Q_{\alpha \beta A} = 0,
\end{equation}
\begin{equation}\label{two}
2\widehat{\delta R}_{\alpha\beta} = \tilde\nabla^{A}Q_{\alpha\beta A} = 0.
\end{equation}
By definition, the two objects also satisfy the relation
\begin{equation}\label{three}
\epsilon^{AB}\tilde\nabla_{B}\left(r^{-2}Q_{\alpha \beta A}\right) - r^{-3}\mathring\nabla_{\alpha} P_{\beta} - r^{-3}\mathring\nabla_{\beta} P_{\alpha} = 0.
\end{equation}

\subsection{Decoupled Quantities}
Applying $r\epsilon^{AB}\tilde\nabla_{B}$ to \eqref{one} and $r\mathring\nabla^{\beta}$ to \eqref{three}, we find
\begin{align*}
& r\epsilon^{AB}\tilde\nabla_{B}\left(r^{-2}\epsilon_{AC}\tilde\nabla^{C}(rP_{\alpha})\right) + r\mathring\nabla^{\beta}\epsilon^{AB}\tilde\nabla_{B}(r^{-2}Q_{\alpha\beta A}) = 0,\\
& -r\mathring\nabla^{\beta}\epsilon^{AB}\tilde\nabla_{B}(r^{-2}Q_{\alpha\beta A}) + r^{-2}\mathring\nabla^{\beta}\mathring\nabla_{\alpha}P_{\beta} + r^{-2}\mathring\nabla^{\beta}\mathring\nabla_{\beta}P_{\alpha} = 0.
\end{align*}

Adding the two, we have
\[rg^{BC}\tilde\nabla_{B}\left(r^{-2}\tilde\nabla_{C}(rP_{\alpha})\right)+ r^{-2}P_{\alpha} + r^{-2}\mathring\Delta P_{\alpha} = 0,\]
decoupling $P$.

Expanding the first term above, we find
\[ \tilde\Box P_{\alpha} + r^{-2}\mathring\Delta P_{\alpha} + \left(r^{-2} - 2r^{-2}r^{B}r_{B} + r^{-1}\left(\tilde\Box r\right)\right)P_{\alpha} = 0.\]
Applying the formula for the spin-$1$ d'Alembertian \eqref{CovectorWave}, along with the background formulae \eqref{miscellaneous}, we arrive at the Regge-Wheeler type equation
\begin{equation}\label{RW1}
\slashed{\Box}_{\mathcal{L}(-1)} P = W P,
\end{equation}
with potential
\begin{equation}
W := \frac{1}{r^2}\left(1 - \frac{8M}{r}\right).
\end{equation}
We refer to the co-vector $P$ as the Cunningham-Moncrief-Price function, following the work \cite{CMP}; see also \cite{MartelPoisson}.

Next, we act on \eqref{one} by the operator $\mathring\nabla_{\beta}$.  Multiplying \eqref{three} by $r^{4}$, and applying the operator $r^{-2}\epsilon_{AB}\tilde\nabla^{B}$ to the result, we find
\begin{align*}
&  r^{-2}\epsilon_{AB}\tilde\nabla^{B}(r \mathring\nabla_{\beta}P_{\alpha}) + r^{-2}\mathring\nabla_{\beta}\mathring\nabla^{\gamma}Q_{\alpha \gamma A} = 0,\\
& r^{-2}\epsilon_{AB}\tilde\nabla^{B}\left(r^4\epsilon^{CD}\tilde\nabla_{D}(r^{-2}Q_{\alpha\beta C})\right)\\
&- r^{-2}\epsilon_{AB}\tilde\nabla^{B}(r\mathring\nabla_{\alpha}P_{\beta}) - r^{-2}\epsilon_{AB}\tilde\nabla^{B}(r\mathring\nabla_{\beta}P_{\alpha})=0.
\end{align*}
Symmetrizing the first equation and summing, we find
\begin{equation}\label{Qeqn}
\begin{split}
&r^{-2}\epsilon_{AB}\tilde\nabla^{B}\left(r^4\epsilon^{CD}\tilde\nabla_{D}(r^{-2}Q_{\alpha\beta C})\right)\\
&+ r^{-2}\mathring\nabla_{\beta}\mathring\nabla^{\gamma}Q_{\alpha\gamma A} + r^{-2}\mathring\nabla_{\alpha}\mathring\nabla^{\gamma}Q_{\beta\gamma A} = 0.
\end{split}
\end{equation}

The first term above can be expanded by appealing to the relation
\[ \epsilon_{AB}\epsilon^{CD}P^{B}_{DC} = P^{B}_{BA} - P^{B}_{AB},\]
valid for tensors on the two-dimensional quotient space.  Applying this result, and contracting the equation with $r^{A}$, we find
\begin{align*}
&\left(\tilde\Box Q_{\alpha\beta A}\right)r^{A} - 2r^{-1}r^{A}r^{B}\tilde\nabla_{B}Q_{\alpha\beta A} - r^{-1}\left(\tilde\Box r\right)Q_{\alpha\beta A}r^{A}\\
&+2r^{-1}(r^{B}r_{B})\tilde\nabla^{B}Q_{\alpha\beta B} - r^{A}\left(\tilde\nabla^{B}\tilde\nabla_{A} Q_{\alpha\beta B}\right)\\
&+r^{-2}\mathring\nabla_{\alpha}\mathring\nabla^{\gamma}\left(Q_{\beta\gamma A}r^{A}\right) + r^{-2}\mathring\nabla_{\beta}\mathring\nabla^{\gamma}\left(Q_{\alpha\gamma A}r^{A}\right) = 0.
\end{align*}

Commuting the covariant derivative, and applying \eqref{two}, we rewrite the term
\[r^{A}\left(\tilde\nabla^{B}\tilde\nabla_{A}Q_{\alpha\beta B}\right) = \tilde{K}Q_{\alpha\beta A}r^{A}.\]

With this and application of \eqref{two} to the other divergence term, our equation takes the form
\begin{align*}
&\left(\tilde\Box Q_{\alpha\beta A}\right)r^{A} - 2r^{-1}r^{A}r^{B}\tilde\nabla_{B}Q_{\alpha\beta B} - r^{-1}\left(\tilde\Box r\right)Q_{\alpha\beta A}r^{A}\\
&- \tilde{K}Q_{\alpha\beta A}r^{A} +r^{-2}\left(\mathring\nabla_{\alpha}\mathring\nabla^{\gamma}Q_{\beta\gamma A}r^{A}\right) + r^{-2}\left(\mathring\nabla_{\beta}\mathring\nabla^{\gamma}Q_{\alpha\gamma A}r^{A}\right) = 0.
\end{align*}

Comparing this expression with the spin-$2$ d'Alembertian \eqref{TwoTensorWave} applied to $Q_{\alpha\beta A}r^{A}$,
\begin{align*}
&\slashed{\Box}_{\mathcal{L}(-2)}\left(Q_{\alpha\beta A}r^{A}\right) = \left(\tilde\Box Q_{\alpha\beta A}\right)r^{A} + \left(\tilde\Box r^{A}\right)Q_{\alpha\beta A}\\
&+2\left(\tilde\nabla^{A}\tilde\nabla^{B} r\right)\tilde\nabla_{B}Q_{\alpha\beta A} - 2r^{-1}r^{A}\left(\tilde\nabla_{A}\tilde\nabla^{B} r\right)Q_{\alpha\beta B}\\
&-2r^{-1}r^{A}r^{B}\tilde\nabla_{A}Q_{\alpha\beta B} + r^{-2}\mathring\Delta\left(Q_{\alpha\beta A}r^{A}\right)\\
&+2r^{-2}\left(r^{B}r_{B}\right)\left(Q_{\alpha\beta A}r^{A}\right) - 2r^{-1}\left(\tilde\Box r\right)\left(Q_{\alpha\beta A}r^{A}\right),
\end{align*}
a lengthy reduction, using the linearized Einstein equation \eqref{two}, the background calculations \eqref{miscellaneous}, and the commutation relation
\[ \mathring\nabla_{\alpha}\mathring\nabla^{\gamma}Q_{\beta\gamma A} + \mathring\nabla_{\beta}\mathring\nabla^{\gamma}Q_{\alpha \gamma A} - \mathring\Delta Q_{\alpha\beta A} = -2Q_{\alpha\beta A},\]
yields the equation
\begin{equation}\label{RW2}
\slashed\Box_{\mathcal{L}(-2)}\left(Q_{\alpha\beta A}r^{A}\right) = V^{(-)}\left(Q_{\alpha\beta A}r^{A}\right),
\end{equation}
with
\begin{equation}\label{RWPotential}
V^{(-)} := \frac{4}{r^2}\left(1-\frac{2M}{r}\right).
\end{equation}

Subsequently, we denote $Q^{(-)}_{\alpha\beta} := Q_{\alpha\beta A}r^{A}$, referred to as the Regge-Wheeler function per Martel-Poisson \cite{MartelPoisson}.  We further denote
\begin{equation}
Q_{t} := Q_{\alpha\beta A}T^{A}.
\end{equation}

\subsection{Spin-Raising of $P_{\alpha}$}

We denote by $\mathcal{D}$ the symmetrized gradient operation, and consider the quantity
\[ S_{\alpha\beta} := r\left(\mathcal{D}P\right)_{\alpha\beta} = r\left(\mathring\nabla_{\alpha}P_{\beta} + \mathring\nabla_{\beta}P_{\alpha}\right).\]

Expanding with definition \eqref{TwoTensorWave}, we find
\begin{align*}
\slashed\Box_{\mathcal{L}(-2)}S_{\alpha\beta} &= \tilde\Box S_{\alpha\beta} -2r^{-1}r^{A}\tilde\nabla_{A}S_{\alpha\beta} \\
&+r^{-2}\mathring\Delta S_{\alpha\beta} + 2r^{-2}r^{A}r_{A}S_{\alpha\beta} -2r^{-1}\left(\tilde\Box r\right)S_{\alpha\beta}\\
&=\mathcal{D}\left(\left(\tilde\Box r\right)P_{\alpha} + r\tilde\Box P_{\alpha} + 2r^{A}\tilde\nabla_{A} P_{\alpha}\right)\\
&+\mathcal{D}\left(-2r^{A}\tilde\nabla_{A}P_{\alpha} -2r^{-1}r^{A}r_{A}P_{\alpha}\right)\\
&+r^{-1}\left(\mathcal{D}\mathring\Delta P_{\alpha} + 3\mathcal{D}P_{\alpha}\right)\\
&+2r^{-2}r^{A}r_{A}S_{\alpha\beta} -2r^{-1}\left(\tilde\Box r\right)S_{\alpha\beta},
\end{align*}
where we have used
\[ \mathring\Delta \mathcal{D} P = \mathcal{D}\mathring\Delta P + 3\mathcal{D} P.\]

Grouping terms and applying the definition \eqref{CovectorWave}, we find
\begin{align*}
\slashed\Box_{\mathcal{L}(-2)} S_{\alpha\beta} &= r\mathcal{D}\slashed\Box_{\mathcal{L}(-1)}P + 3r^{-2}S_{\alpha\beta}\\
&= \left(W + 3r^{-2}\right)S_{\alpha\beta},
\end{align*}
where we have used the wave equation \eqref{RW1}.  That is, the spin-raised quantity $S_{\alpha\beta}$ satisfies
\begin{equation}
\slashed\Box_{\mathcal{L}(-2)} S_{\alpha\beta} = V^{(-)} S_{\alpha\beta},
\end{equation}
with $V^{(-)}$ defined by \eqref{RWPotential}.

\subsection{Master Quantity for the Co-Closed Portion} \label{master_coclosed} 
The master quantity $Q^{(-)}$ can be rewritten in an alternate form to facilitate comparison with the master quantity \eqref{closed} in the closed portion. For a co-closed solution $h_2$ \eqref{h_2}, we define
\begin{equation}
\underline{\epsilon}_{A} := \underline {H}_{A} - \frac{1}{2}r^2\tilde\nabla_{A}(r^{-2}\underline{H}_{2}).
\end{equation} 
It is not hard to see that 
\begin{equation}\label{co-closed} Q^{(-)}_{\alpha\beta} =\left(\epsilon_{\alpha}^{\gamma}\mathring\nabla_{\beta} + \epsilon_{\beta}^{\gamma}\mathring\nabla_{\alpha}\right)\mathring\nabla_{\gamma} (r^A \underline{\epsilon}_A).\end{equation}

It can be shown that $r^A\underline{\epsilon}_{A}$ satisfies a Regge-Wheeler type equation with respect to the operator $\tilde{\Box}$. However, the subsequent decay estimates in this article are obtained after identifying $r^A\underline{\epsilon}_A$ through \eqref{co-closed} as a section of the bundle of symmetric traceless two-tensors. 

\subsection{Relation to Chandrasekhar}
Following Chandrasekhar \cite{Chandra1}, in the axial case,
\begin{equation}
h_{ab}dx^{a}dx^{b} = -2r^{2}\sin^2\theta(\omega dt d\phi + q_2 dr d\phi + q_3 d\theta d\phi).
\end{equation}

Using the definitions \eqref{Pdef} and \eqref{Qdef}, we calculate
\begin{align}
P &= r^3\sin^2\theta Q_{02}d\phi,\\
Q &= \sin^2\theta Q_{03} (r^{-2}d\theta d\phi) dt + \frac{\Delta}{r^2}\sin^2\theta Q_{23} (r^{-2}d\theta d\phi) dr_{*}.
\end{align}

The decoupled quantities $P$ and $Q^{(-)}$ correspond, respectively, to $\beta$ and $\alpha$ in the first two authors' earlier work on axial perturbations \cite{HK}.  Indeed, the decoupling procedures are the same as those in the previous work.  In contrast to \cite{HK}, however, the above decoupling is achieved without the need for axisymmetry of the co-closed solution.

\subsection{Analysis of the Co-Closed Solution}

\subsubsection{Decay of $P$, $Q^{(-)}$, and $Q_{t}$}

The analysis of the Cunningham-Moncrief-Price function $P$ and the Regge-Wheeler function $Q^{(-)}$ via their equations, together with the derived estimates on $Q_{t}$ via the linearized vacuum Einstein equations, is identical to that carried out in the earlier work \cite{HK}; although the results of \cite{HK} apply to axisymmetric perturbations, in accordance with the framework laid out in \cite{Chandra1}, the associated analysis does not depend upon axisymmetry.  Estimates for the Regge-Wheeler equation appear in the earlier papers \cite{BlueSoffer, DSS, FriedmanMorris, DHR}; the novel feature of the present work is in exploiting these estimates for both the Cunningham-Moncrief-Price function and the Regge-Wheeler function to deduce estimates on $Q_{t}$, with estimates on all three leading to control of the linearized metric following a suitable gauge normalization.  After listing relevant notation from \cite{HK}, we collect the main results below.

Decay is expressed with respect to the foliation of smooth spacelike hypersurfaces $\tilde{\Sigma}_{\tau}$, characterized by
\begin{align}\label{decayFoliation}
\begin{split}
\tau &= t+2M\log(r-2M)+c_0\textup{, for}\ r\leq 3M,\\
 &= t-\sqrt{r^2+1}+c_1\textup{, for}\ r \geq 20M,
\end{split}
\end{align}
with the specification in the spatially precompact region $3M < r < 20M$ and the choice of constants $c_0$ and $c_1$ made smoothly in such a way that $u,v\geq \tau$ on $\tilde{\Sigma}_\tau$.  

Energies are specified in terms of the red-shift multiplier $N$, first identified in \cite{DR}.  For the purposes of the energy calculations below, it suffices to note that $N$ is a future-directed, strictly timelike commutator yielding a positive-definite energy
\begin{equation}\label{NenergyXi}
E^{N}_{\xi}(\tilde{\Sigma}_{\tau}) := \int_{\tilde{\Sigma}_{\tau}} J^{N}_{a}[\xi]\eta^{a}_{\tilde{\Sigma}_{\tau}}
\end{equation}
for a symmetric traceless two-tensor $\xi$ on the folation $\tilde{\Sigma}_{\tau}$.  More specific details on the red-shift multiplier are presented in Section \ref{RedShift}.
We also make use of the initial energies
\begin{align}
E_0[\xi] &:= \sum_{(q)\leq 2}\int_{\{t = 0 \}} J^{N}_{a}[\mathcal{K}^{(q)}\xi]\eta^{a}\label{E0def},\\
E_1[\xi] &:= \sum_{(q)\leq 4} \int_{\{t = 0 \}} (1+r_*^2)J^{N}_{a}[\mathcal{K}^{(q)}\xi]\eta^{a}\label{E1def},\\
E_2[\xi] &:= \sum_{(q)\leq 6} \int_{\{t = 0 \}} (1+r_*^2)J^{N}_{a}[\mathcal{K}^{(q)}\xi]\eta^{a}\label{E2def},
\end{align}
defined for the same $\xi$ on the time slice $\{ t = 0 \}$.  Note that the sum is taken over multi-indices $(q)$ of length $q$ and less, over all Killing commutators in $\mathcal{K}$.  Finally, we use $|\ |_{g}$ to denote the spacetime norm; we remind the reader that the spacetime norm is positive-definite on the relevant sphere bundles.

Analysis of the decoupled quantities $P$ and $Q^{(-)}$  yields the following decay estimates:
\begin{theorem}\label{decay_r_large}
Suppose $P$ and $Q$ are defined as in \eqref{Pdef} and \eqref{Qdef}, respectively, and satisfy the linearized vacuum Einstein equations \eqref{one} and \eqref{two}.  Owing to the decoupling procedure above, $P$ and $Q^{(-)}$ satisfy the Regge-Wheeler type equations \eqref{RW1} and \eqref{RW2}.  Assume that $P$ and $Q$ are smooth and compactly supported on $\{ t = 0 \}$, with support on $\ell \geq 2$.  Then we have the decay estimates

\begin{align*}
|Q^{(-)}|_{g}&\leq C\sqrt{E_2[Q^{(-)}]}r^{-1}\tau^{-1/2},\\
|Q^{(-)}|_{g}&\leq C\sqrt{E_2 [Q^{(-)}]}\tau^{-1},\\
|P|_{g}&\leq C\sqrt{E_2[P]}r^{-1}\tau^{-1/2},\\
|P|_{g}&\leq C\sqrt{E_2 [P]}\tau^{-1},
\end{align*}
on the family of hypersurfaces $\tilde{\Sigma}_{\tau}$ \eqref{decayFoliation}.
\end{theorem}

Near the horizon, it is possible to apply the transverse direction $\hat{Y}:= \frac{1}{1-\mu}(\partial_{t} - \partial_{r_{*}})$ as a commutator to obtain further decay estimates for $P$.  The multiplier is not Killing, but the error terms stemming from the calculation
\begin{align}
\begin{split}
&\slashed{\Box}_{\mathcal{L}(-1)}(\slashed\nabla_{\hat{Y}}P) = \frac{2(r-M)}{r^2}\slashed\nabla_{\hat{Y}}\slashed\nabla_{\hat{Y}}P - \frac{4}{r}\slashed\nabla_{\hat{Y}}\slashed\nabla_{T}P\\
&+ \frac{2}{r^2}(\slashed{\nabla}_{T}-\slashed\nabla_{\hat{Y}})P + (\hat{Y}W-\frac{2}{r}W)P.
\end{split}
\end{align}
are controllable, in much the same way as in Section 3.3.4 of \cite{DRClay}.  As a consequence, we have the following theorem:
\begin{theorem}\label{decay_r_small}
Suppose that $P$ is a solution of \eqref{RW1}, smooth and compactly supported on $\{ t = 0 \}$, with support on $\ell \geq 2$.  Fixing a sufficiently small radius $r_1$ (see the red-shift construction of Section \ref{RedShift}), we have the decay estimate
\begin{equation}
\sup_{\tilde{\Sigma}_{\tau} \cap \{r\leq r_1\}}|\slashed\nabla_{\hat{Y}}P|_{g}\leq C\left(\sqrt{E_2 [P]} + \sqrt{E_2[\slashed{\nabla}_{\hat{Y}}P]}\right)\tau^{-1},
\end{equation}
where $\hat{Y}$ is the transverse direction \eqref{hatY}.
\end{theorem}

Using the estimates of Theorem \ref{decay_r_large}, we are also able to deduce pointwise decay of the remaining component $Q_{t}$.  As these estimates arise from a direct analysis of the linearized vacuum Einstein equations, rather than that of a decoupled Regge-Wheeler type equation, we present the details of the argument below.  Throughout, we make use of notation arising in the vector-field multiplier method; we refer the reader to Section 8 for further details.

\begin{theorem}\label{QtEstimates}
Suppose $P$ and $Q$ are defined as in \eqref{Pdef} and \eqref{Qdef}, respectively, and satisfy the linearized vacuum Einstein equations \eqref{one} and \eqref{two}.  Assume, moreover, that $P$ and $Q$ are smooth and compactly supported on $\{ t = 0 \}$, with support on $\ell \geq 2$.  Then the component $Q_{t}$ satisfies the decay estimates
\begin{equation}
\begin{split}
|Q_{t}|_{g} &\leq C\left(\sqrt{E_2[Q^{(-)}]}+\sqrt{E_2[P]}\right)r^{-1}\tau^{-1/2},\\
|Q_{t}|_{g}&\leq C\left(\sqrt{E_2 [Q^{(-)}]}+\sqrt{E_2 [P]}\right)\tau^{-1}.
\end{split}
\end{equation}
\end{theorem}

\begin{proof}
Recall \eqref{two} and \eqref{three}:
\begin{align*}
&2\widehat{\delta R}_{\alpha \beta} = \tilde\nabla^{A}Q_{\alpha\beta A} = 0,\\
&\epsilon^{AB}\tilde\nabla_{B}\left(r^{-2}Q_{\alpha \beta A}\right) - 2r^{-3}\mathring\nabla_{(\alpha} P_{\beta)}= 0.
\end{align*}

Expanding and rewriting in terms of the projected covariant derivative $\slashed{\nabla},$ we have the relations
\begin{align}
&\slashed{\nabla}_{t}Q_{\alpha\beta t} = \slashed{\nabla}_{r_{*}}Q_{\alpha\beta r_{*}} + \frac{2\Delta}{r^3}Q_{\alpha\beta r_{*}},\\
&\slashed{\nabla}_{r_{*}}Q_{\alpha\beta t} = \slashed{\nabla}_{t}Q_{\alpha\beta r_{*}} + 2\left(1-\frac{2M}{r}\right)r^{-1}\slashed\nabla_{(\alpha} P_{\beta)}.
\end{align}

Let $\rho$ be a geodesic radial coordinate on $\tilde{\Sigma}_{\tau}$, normalized to $\rho = 1$ on the horizon.  Further, denote by $\eta_{\tilde{\Sigma}_{\tau}}$ the unit normal vector of $\tilde{\Sigma}_{\tau}$.  

Near the horizon, 
\begin{align*}
&\frac{\partial}{\partial \rho} \sim -\frac{2Mr}{\Delta}\partial_{t} + \frac{r^2}{\Delta}\partial_{r_{*}},\\
&\eta_{\tilde{\Sigma}_{\tau}} \sim \frac{r^2}{\Delta}\partial_{t} -\frac{2Mr}{\Delta} \partial_{r_{*}},
\end{align*}
such that
\begin{align*}
\slashed{\nabla}_{\rho}Q_{\alpha\beta t} &\sim \slashed{\nabla}_{\eta_{\tilde{\Sigma}_{\tau}}}Q_{\alpha\beta r_{*}} -\frac{4M}{r^2}Q_{\alpha\beta r_{*}} + 2r^{-1}\slashed\nabla_{(\alpha} P_{\beta)},\\
|\slashed{\nabla}_{\rho}Q_{t}|^2 &\leq C\left(|\slashed{\nabla}_{\eta_{\tilde{\Sigma}_{\tau}}}Q_{r_{*}}|^2 + |Q_{r_{*}}|^2 + |\tilde{\slashed{\nabla}}P|^2\right)\\
&\leq C\left(J^{N}_{a}[Q_{r_{*}}]\eta^{a}_{\tilde{\Sigma}_{\tau}} + J^{N}_{a}[P]\eta^{a}_{\tilde{\Sigma}_{\tau}}\right).
\end{align*}

Near null infinity, 
\begin{align*}
&\frac{\partial}{\partial\rho}  \sim \sqrt{r^2+1}\partial_r+r\partial_t,\\
&\eta_{\tilde{\Sigma}_{\tau}}\sim r\partial_r+\sqrt{r^2+1}\partial_t,
\end{align*}
such that
\begin{align*}
\slashed{\nabla}_\rho Q_{\alpha\beta t} &\sim \slashed{\nabla}_{\eta_{\tilde{\Sigma}_\tau}} Q_{\alpha\beta r_*}+2 Q_{\alpha\beta r_*}+2\slashed{\nabla}_{(\alpha}P_{\beta )},\\
\left|\slashed{\nabla}_{\rho} Q_{t}\right|^2&\leq C \left( |\slashed{\nabla}_{\eta_{\tilde{\Sigma}_\tau}} Q_{r_*}|^2+|Q_{ r_*}|^2+|\tilde{\slashed{\nabla}}P|^2 \right)\\
&\leq C r(J^N_a [Q_{r_*}]\eta^a_{\tilde{\Sigma}_\tau}+J^N_a [P]\eta^a_{\tilde{\Sigma}_\tau}).
\end{align*}

As the governing equations commute with the angular Killing fields $\Omega_{i}$, the estimates above also hold for $\Omega^{(q)}Q_{t}$, $\Omega^{(q)}Q_{r_{*}}$, $\Omega^{(q)}P$, with $(q)$ a multi-index of length $q$.

Integrating, we deduce the first decay estimate
\begin{align*}
|Q_{t}|\leq &\int_1^\infty |\slashed{\nabla}_{\rho}Q_{t}|d\rho \leq C \left(\int_1^\infty |\slashed{\nabla}_{\rho}Q_{t}|^2\rho^2 d\rho\right)^{1/2}\\
     \leq &C\left( \int_1^\infty d\rho\int_{S^2}d\sigma \rho^2\left( |\slashed{\nabla}_{\rho}Q_{t}|^2+|\slashed{\nabla}_{\rho}\Omega Q_{t}|^2+|\slashed{\nabla}_{\rho}\Omega^2Q_{t}|^2 \right) \right)^{1/2}\\
     \leq &C\left(\int_{\tilde{\Sigma}_{\tau}} \frac{\rho^2}{r^2}(|\slashed{\nabla}_{\eta_{\tilde{\Sigma}_{\tau}}}Q_{r_{*}}|^2 + |Q_{r_{*}}|^2 + |\tilde{\slashed{\nabla}}P|^2 + \hdots)\right)^{1/2}\\
     \leq &C\left(\sqrt{E_2[Q_{r_{*}}]} + \sqrt{E_2[P]}\right)\tau^{-1},
\end{align*}
where we have used the comparison $r \sim e^{\rho}$ near null infinity.

Regarding the second estimate, we restrict our attention to the region $r\geq 3M$, and write
\begin{align*}
\int_{S^2(r)}|Q_{\alpha\beta t}|_g^2d \mathring{\sigma}&=\int_{S^2(r)}r^{-4}|Q_{\alpha\beta t}|^2_{\mathring{\sigma}} d\mathring{\sigma}\\
&\leq C\int_{S^2(r)}r^{-4}|\mathring{\nabla}^{\beta} Q_{\alpha\beta t}|^2_{\mathring{\sigma}} d\mathring{\sigma}\\
&= C\int_{S^2(r)}r^{-2}|\mathring{\nabla}^{\beta} Q_{\alpha\beta t}|^2_{g} d\mathring{\sigma}\\
&\leq C\int_{S^2 (r)}r^{-2}\left(r^2|\slashed{\nabla}_r P|^2+|P|^2\right)d\mathring{\sigma},
\end{align*}
where we have used the equation \eqref{one}.  From $|P|^2\leq CE_2[P]r^{-2}\tau^{-1}$ and commutation with $T$, we have $|\slashed{\nabla}_t\slashed{\nabla}_tP|^2\leq CE_2[T^2P]r^{-2}\tau^{-1}$.  The vector field $T$ is strictly timelike for $r\geq 3M$, and the equation \eqref{RW1} gives $|\slashed{\Delta}_{\Sigma_t} P|^2\leq C\left(E_2[P] + E_2[T^2P]\right)r^{-2}\tau^{-1}$, where $\Sigma_t$ are the constant time $t$-slices.  This last estimate, together with $|P|^2\leq CE_2[P]r^{-2}\tau^{-1}$, yields a $W^{2,q}$ estimate for $P$ in a unit ball on $\Sigma_t$, for any $1<q<\infty$. Using the Sobolev embedding $W^{2,q}\subset C^1$ for $q>3$, we can estimate the term $|\slashed{\nabla}_r P|^2$ in the integral above, from which the estimate on $Q_{t}$ follows.

\end{proof}

Note that the quantities $P$, $Q^{(-)},$ and $Q_{t}$ also exhibit uniform decay in energy, as well as uniform boundedness, both pointwise and energy; see \cite{HK}.  For the purposes of pulling back decay estimates at the linearized metric level, it suffices to consider just the pointwise decay estimates above.

\subsubsection{The Regge-Wheeler Gauge}
To derive decay estimates on the linearized metric level, we impose the Regge-Wheeler gauge.
\begin{lemma}\label{RWcovector}
For any $h_2$ of the co-closed form \eqref{h_2}, there is a co-vector $G$ with $\pi_{G}$ of the co-closed form \eqref{pi_2} such that $h=h_2-\pi_{G}$ satisfies
\[ h_{\alpha\beta}= \hat{h}_{\alpha\beta} = 0.\]
\end{lemma}
\begin{proof}
Comparing the equations \eqref{h_2} and \eqref{pi_2}, we see that the co-vector field $G$ is easily constructible by the choice $\underline{G}_2 = \frac{1}{2}\underline{H}_2.$  Indeed, such a choice exhausts our gauge freedom.
\end{proof}

\subsubsection{Decay of the Linearized Solution}

With the imposition of the Regge-Wheeler gauge, we are ready to prove the main theorem of this section, on the decay of the co-closed solution.  With the gauge-normalized metric, 
\[ h = h_{A\alpha}dx^{A}dx^{\alpha},\]
we prove decay of the linearized metric components in the following sense.  First we specify the frames $\{\hat{Y}^{A},T^{A}\}$ and $\{T^{A}, r^{A}\}$ on the quotient space $\mathcal{Q}$, noting that the frames are regular near the event horizon and away from the event horizon, respectively.  Contracting with the linearized metric, we prove decay for the quantities
\[ \hat{Y}^{A}h_{A\alpha},\ T^{A}h_{A\alpha},\ r^{A}h_{A\alpha},\]
understood as co-vectors on the quotient spheres.  We remind the reader that the spacetime norm $|\ |_{g}$ is positive-definite on the sphere bundle in question.
\begin{theorem}\label{coclosedMain}
Suppose $h_2$ is a co-closed solution \eqref{h_2} of the linearized vacuum Einstein equations \eqref{linearized_Einstein}, with support in $\ell \geq 2$.  Further, assume that $h_2$ is smooth and compactly supported, with normalization
\[h = h_2 - \pi_{G}\]
in the Regge-Wheeler gauge.  Then the linearized metric components of the normalized solution $h$ satisfy the spacetime decay estimates
\begin{align}\label{metricDecay1}
\begin{split}
&\sup_{\tilde{\Sigma}_{\tau}} |r^{A}h_{A\alpha}|_{g}  \leq C\Big( \sum_{(q)\leq 2} \sqrt{E_2[\Omega^{(q)} Q^{(-)}]}\Big)\tau^{-1/2},\\
&\sup_{\tilde{\Sigma}_{\tau}} |T^{A}h_{A\alpha}|_{g} \leq C\Big( \sum_{(q)\leq 2} \Big(\sqrt{E_2[\Omega^{(q)} Q^{(-)}]}+\sqrt{E_2[\Omega^{(q)} P]}\Big)\Big)\tau^{-1/2},\\
&\sup_{\tilde{\Sigma}_{\tau}\cap{\{r \leq r_1\}}} |\hat{Y}^{A}h_{A\alpha}|_{g} \\
&\leq C\Big( \sum_{(q)\leq 2} \Big(\sqrt{E_2[\Omega^{(q)} P]} +\sqrt{E_2[\Omega^{(q)}\slashed{\nabla}_{\hat{Y}}P]}\Big)\Big)\tau^{-1/2},
\end{split}
\end{align}
through the decay foliation \eqref{decayFoliation}.

Moreover, in the radially compact region $2M \leq r \leq R_0$, we have the improvement
\begin{align}\label{metricDecay2}
\begin{split}
&\sup_{\tilde{\Sigma}_{\tau}\cap{\{r\leq R_0\}}} |r^{A}h_{A\alpha}|_{g}  \leq C(R_0)\Big( \sum_{(q)\leq 2} \sqrt{E_2[\Omega^{(q)} Q^{(-)}]}\Big)\tau^{-1},\\
&\sup_{\tilde{\Sigma}_{\tau}\cap{\{r\leq R_0\}}} |T^{A}h_{A\alpha}|_{g} \\
&\leq C(R_0)\Big( \sum_{(q)\leq 2} \Big(\sqrt{E_2[\Omega^{(q)} Q^{(-)}]}+\sqrt{E_2[\Omega^{(q)} P]}\Big)\Big)\tau^{-1}.
\end{split}
\end{align}

\end{theorem}
\begin{proof}

The normalized solution $h = h_{A\alpha}dx^{A}dx^{\alpha}$ is related to the gauge-invariant quantity $Q_{\alpha\beta A}$ \eqref{Qdef} by
\[Q_{\alpha\beta A}dx^{\alpha}dx^{\beta}dx^{A}  = \left(\mathring\nabla_{\beta} h_{A\alpha} + \mathring\nabla_{\alpha} h_{A\beta}\right)dx^{\alpha}dx^{\beta}dx^{A}.\]  
We estimate
\begin{align*}
\||r^{A}h_{A\alpha}|_{g}\|_{L^2(S^2)}^2&= \int_{S^2} |r^{A}h_{A\alpha}|_g^2 \\
&\leq C\int_{S^2} \left(r^{-2}\mathring\sigma^{\alpha\beta}\mathring\sigma^{\delta\gamma}Q^{(-)}_{\alpha\delta}Q^{(-)}_{\beta\gamma}\right)\\
&\leq Cr^2\||Q^{(-)}|_{g}\|_{L^2(S^2)}^2\\
&\leq C (E_2 [Q^{(-)}])\tau^{-1},
\end{align*}
where we have used a Poincar\'{e} inequality on the sphere and Theorem \ref{decay_r_large} to absorb the large radial weight.  A similar argument gives
the estimate
\[\||T^{A}h_{A\alpha}|_{g}\|_{L^2(S^2)}^2 \leq Cr^2\||Q_{t}|_{g}\|_{L^2(S^2)}^2\leq C (E_2 [Q^{(-)}]+E_2 [P])\tau^{-1}.\]
For the remaining quantity, we have
\[\||\hat{Y}^{A}h_{A\alpha}|_{g}\|_{L^2(S^2)}^2 \leq Cr^2\||\hat{Y}^{A}Q_{\alpha\beta A}|_{g}\|_{L^2(S^2)}^2,\]
as before.  We then use the linearized equation \eqref{one} to relate $\hat{Y}^{A}Q_{\alpha\beta A}$ with $\slashed{\nabla}_{\hat{Y}}P_{\alpha}$ and Theorem \ref{decay_r_small} to deduce the decay estimate with $r \leq r_1$:
\begin{align*}
\||\hat{Y}^{A}h_{A\alpha}|_{g}\|_{L^2(S^2)}^2 &\leq C\||\slashed{\nabla}_{\hat{Y}}P_{\alpha}|_{g}\|_{L^2(S^2)}^2\\ &\leq C\left(E_2 [P] + E_2[\slashed{\nabla}_{\hat{Y}}P]\right)\tau^{-2}
\end{align*}

Restricting to the radially compact region $2M \leq r \leq R_0$, we absorb radial terms into the constant factor and apply Theorem \ref{decay_r_large} to deduce
\begin{align*}
\||r^{A}h_{A\alpha}|_{g}\|_{L^2(S^2)}^2 &\leq Cr^2\||Q^{(-)}|_{g}\|_{L^2(S^2)}^2 \leq C(R_0) (E_2 [Q^{(-)}])\tau^{-2},\\
\||T^{A}h_{A\alpha}|_{g}\|_{L^2(S^2)}^2 &\leq C(R_0) (E_2 [Q^{(-)}]+E_2 [P])\tau^{-2}.
\end{align*}

Commuting with the angular Killing vector fields $\Omega_{i}$ and applying the Sobolev inequality on the sphere, we obtain the desired pointwise bounds. 
\end{proof}

\section{Decoupling of the Zerilli-Moncrief Function}

The remainder of the paper concerns the analysis of the closed solution $h_1$ \eqref{h_1} of $\delta g^{\ell \geq 2}$, primarily accomplished by study of the Zerilli-Moncrief function and the Zerilli equation.  In this section, we define the Moncrief-Zerilli function and the associated Zerilli equation, largely following the treatment of Martel-Poisson \cite{MartelPoisson}.

To begin, we decompose the closed solution into spherical harmonics, with the notation
\begin{align}
\begin{split}\label{projection1}
&h_{AB}^{\ell m} = \tilde{h}_{AB}^{\ell m}Y^{\ell m},\\
&H_{A}^{\ell m} = \tilde{H}_{A}^{\ell m}Y^{\ell m},\\
&H^{\ell m} = \tilde{H}^{\ell m}Y^{\ell m},\\
&H_{2}^{\ell m} = \tilde{H}_2^{\ell m}Y^{\ell m},
\end{split}
\end{align}
with $\tilde{h}_{AB}^{\ell m}, \tilde{H}_{A}^{\ell m}, \tilde{H}^{\ell m},$ and $\tilde{H}_2^{\ell m}$ all objects on the quotient space $\mathcal{Q}$.  The spherical decomposition above is done in $L^2$ on the spheres of symmetry, so that, for example,
$$ H_2(t,r,\theta,\phi) = \sum_{\ell, m} H_2^{\ell m}(t,r,\theta,\phi) = \sum_{\ell, m}  \tilde{H}_2^{\ell m}(t,r)Y^{\ell m}(\theta,\phi)$$
holds for each pair $(t,r)$ in the sense of $L^2(S^2)$.

Defining
\begin{equation}
e_{A}^{\ell m} := \tilde{H}_{A}^{\ell m} - \frac{1}{2}r^2\tilde\nabla_{A}(r^{-2}\tilde{H}_{2}^{\ell m}),
\end{equation}
it is well-known that the quantities
\begin{align}
\begin{split}
&\tilde{k}_{AB}^{\ell m} := \tilde{h}_{AB}^{\ell m} - \tilde\nabla_{A} e_{B}^{\ell m} - \tilde\nabla_{B} e_{A}^{\ell m},\\
&\tilde{K}^{\ell m} := r^{-2}\tilde{H}^{\ell m} + \frac{1}{2}\ell(\ell + 1)r^{-2}\tilde{H}_{2}^{\ell m} -\frac{2}{r}r^{A}e_{A}^{\ell m},
\end{split}
\end{align}
are gauge-invariant under gauge transformations of the closed form \eqref{pi_1}.  

In terms of $\tilde{k}_{AB}^{\ell}$ and $\tilde{K}^{\ell m}$, the gauge-invariant Zerilli-Moncrief function is defined by
\begin{equation}\label{Zdef}
Z^{(+)}_{\ell m} := \frac{2r}{\ell(\ell +1)}\left[\tilde{K}^{\ell m} + \frac{2}{\Lambda}\left(r^{A}r^{B}\tilde{k}_{AB}^{\ell m} -rr^{A}\tilde\nabla_{A}\tilde{K}^{\ell m}\right)\right],
\end{equation}
where
\begin{equation}
\Lambda := (\ell - 1)(\ell + 2) + \frac{6M}{r}.
\end{equation}

For the sake of simplicity, we often suppress $m$-dependence in the spherical harmonics.  Moreover, we encode $\ell$-dependence with the shorthand
\begin{equation}\label{nDefinition}
n := \frac{1}{2}(\ell - 1)(\ell + 2).
\end{equation}

In this notation, the Zerilli-Moncrief function satisfies the well-known Zerilli equation
\begin{equation}\label{ZerilliEqn}
\tilde{\Box} Z^{(+)}_{n} = \tilde{V}^{(+)}_{n} Z^{(+)}_{n},
\end{equation}
with potential
\begin{equation}
\tilde{V}^{(+)}_{n} := \frac{2}{r^3(nr+3M)^2}\left(n^2(n+1)r^3 + 3Mn^2r^2 + 9M^2nr +9M^3\right).
\end{equation}
See the thesis of Martel \cite{Martel} for a derivation of the Zerilli equation for $Z^{(+)}_{n}$.

Rescaling, we define an associated spacetime symmetric traceless two-tensor 
\begin{equation}\label{alphaDef}
Q^{(+)}_{nm} := \left(rZ^{(+)}_{nm}\right)Y_{\alpha\beta}^{\ell m}dx^{\alpha}dx^{\beta}.
\end{equation}

Applying \eqref{TwoTensorWave}, $Q^{(+)}_{n}$ is seen to satisfy the space-time Zerilli equation
\begin{equation}\label{RW3}
\slashed{\Box}_{\mathcal{L}(-2)}Q^{(+)}_{n} = V^{(+)}_{n}Q^{(+)}_{n},
\end{equation}
where
\begin{equation}
V^{(+)}_{n} := \frac{2}{r^2(nr+3M)^2}\left(2n^2r^2 - 2n(2n-3)Mr - 3(2n-3)M^2\right).
\end{equation}
Note that $V^{(+)}_{n} \rightarrow V^{(-)}$, the Regge-Wheeler potential \eqref{RW2}, asymptotically with increasing harmonic number.

This second-order equation on $Q^{(+)}_{n}$, valid for all higher modes $\ell \geq 2$, is our primary tool in analyzing the closed solution.

\subsubsection{Master Quantity in the Closed Portion} 
It is possible to rewrite the master quantity \eqref{alphaDef} for the closed portion in a way that is independent of the harmonic number. 
Let $h_1$ be a closed solution \eqref{h_1}, and consider
\[e_{A} := {H}_{A}- \frac{1}{2}r^2\tilde\nabla_{A}(r^{-2}{H}_{2}),\]
and
\[\begin{split}
&{k}_{AB} := {h}_{AB}- \tilde\nabla_{A} e_{B} - \tilde\nabla_{B} e_{A},\\
&{K}:= r^{-2}{H}- \frac{1}{2}r^{-2}\mathring{\Delta}{H}_{2} -\frac{2}{r}r^{A}e_{A}.
\end{split}\]

In terms of ${k}_{AB}$ and ${K}$, we consider the following gauge-invariant and harmonic-independent quantity
\begin{equation}\label{closed}
{r}{K}+ 2r{\Lambda}^{-1}\left(r^{A}r^{B}{k}_{AB} -rr^{A}\tilde\nabla_{A}{K}\right),
\end{equation}
where ${\Lambda}^{-1}= (-\mathring{\Delta}-2 + \frac{6M}{r})^{-1}$ is interpreted as an integral operator: for functions with higher angular mode $\ell\geq 2$, the operator $-\mathring{\Delta}-2$ is positive and invertible. 

The quantity \eqref{closed} satisfies the Zerilli equation \eqref{ZerilliEqn}. Decay estimates are obtained after recasting the quantity \eqref{alphaDef}, as in Section \ref{master_coclosed}.

\section{Analysis of $Q^{(+)}_{n}$}

In what follows, we produce decay estimates for the $Q^{(+)}_{n}$, using the familiar vector field multiplier method.  See also the work of Johnson \cite{Johnson}.

\subsection{The Stress-Energy Formalism}
Associated with our wave equation is the stress-energy tensor:
\begin{align}
\begin{split}
T_{ab}[Q^{(+)}_{n}] &:= \slashed\nabla_{a}Q^{(+)}_{n}\cdot\slashed\nabla_{b} Q^{(+)}_{n}\\
& - \frac{1}{2}g_{ab}\left(\slashed\nabla^{c}Q^{(+)}_{n}\cdot\slashed\nabla_{c} Q^{(+)}_{n} + V_{n}^{(+)}|Q^{(+)}_{n}|^2\right),
\end{split}
\end{align}
where we emphasize that
\begin{align*}
&\slashed{\nabla}_{a}Q^{(+)}_{n}\cdot\slashed{\nabla}_{b}Q^{(+)}_{n} = g^{\alpha\beta}g^{\gamma\delta}(\slashed\nabla_{a}Q^{(+)}_{n})_{\alpha\gamma}(\slashed\nabla_{b}Q^{(+)}_{n})_{\beta\delta},\\
&|Q^{(+)}_{n}|^2 = g^{\alpha\beta}g^{\gamma\delta}Q^{(+)}_{n\alpha\gamma}Q^{(+)}_{n\beta\delta},
\end{align*}
and
\begin{align*}
\slashed{\nabla}^{c}Q_{n}^{(+)}\cdot\slashed{\nabla}_{c}Q_{n}^{(+)} &= g^{ab}g^{\alpha\beta}g^{\gamma\delta}(\slashed\nabla_{a}Q_{n}^{(+)})_{\alpha\gamma}(\slashed\nabla_{b}Q_{n}^{(+)})_{\beta\delta},\\
&= \slashed{\nabla}^{A}Q_{n}^{(+)}\cdot \slashed{\nabla}_{A}Q_{n}^{(+)}  + |\tilde{\slashed\nabla}Q_{n}^{(+)}|^2,
\end{align*}
with the notation
\begin{equation}\label{angularGradient}
|\tilde{\slashed\nabla}Q_{n}^{(+)}|^2= g^{\eta\nu}g^{\alpha\beta}g^{\gamma\delta}(\slashed\nabla_{\eta}Q_{n}^{(+)})_{\alpha\gamma}(\slashed\nabla_{\nu}Q^{(+)}_{n})_{\beta\delta}
\end{equation}
for the angular gradient.

Applying a vector field multiplier $X^{b}$, we define the energy current
\begin{equation}
J^{X}_{a}[Q^{(+)}_{n}] := T_{ab}[Q^{(+)}_{n}]X^{b}
\end{equation}  
and the density
\begin{equation}
K^{X}[Q^{(+)}_{n}] := \nabla^{a}J^{X}_{a}[Q^{(+)}_{n}] = \nabla^{a}(T_{ab}[Q^{(+)}_{n}]X^{b}).
\end{equation}

As well, we will have occasion to use the weighted energy current
\begin{equation}\label{Jweight}
J^{X,\omega^{X}}_{a}[Q^{(+)}_{n}] := J^{X}_{a}[Q^{(+)}_{n}] + \frac{1}{4}\omega^{X}\nabla_{a}|Q^{(+)}_{n}|^2 -\frac{1}{4}\nabla_{a}\omega^{X}|Q^{(+)}_{n}|^2, 
\end{equation}
with weighted density
\begin{equation}\label{Kweight}
K^{X,\omega^{X}}[Q^{(+)}_{n}] := K^{X}[Q^{(+)}_{n}] + \frac{1}{4}\omega^{X}\Box |Q^{(+)}_{n}|^2 -\frac{1}{4} \Box \omega^{X} |Q^{(+)}_{n}|^2,
\end{equation}
for a suitable scalar weight-function $\omega^{X}$.

The current $J^{X}_{a}[Q^{(+)}_{n}]$ and density $K^{X}[Q^{(+)}_{n}]$ serve as a convenient notation to express the spacetime Stokes' theorem
\begin{equation}
\int_{\mathcal{\partial D}} J^{X}_{a}[Q^{(+)}_{n}]\eta^{a} = \int_{\mathcal{D}} K^{X}[Q^{(+)}_{n}].
\end{equation} 

Likewise,
\begin{equation}
\int_{\mathcal{\partial D}} J^{X,\omega^{X}}_{a}[Q^{(+)}_{n}]\eta^{a} = \int_{\mathcal{D}} K^{X,\omega^{X}}[Q^{(+)}_{n}].
\end{equation} 

The stress-energy tensor $T_{ab}[Q^{(+)}_{n}]$ defined above has non-trivial divergence
\begin{equation}\label{stressDivergence}
\nabla^{a}T_{ab}[Q^{(+)}_{n}] = -\frac{1}{2}\nabla_{b}V^{(+)}_{n}|Q^{(+)}_{n}|^2 + \slashed{\nabla}^{a}Q^{(+)}_{n}[\slashed{\nabla}_{a},\slashed{\nabla}_{b}]Q^{(+)}_{n},
\end{equation}
where we note crucially that the commutator $[\slashed{\nabla}_{a},\slashed{\nabla}_{b}]$ vanishes when contracted with a multiplier invariant under the angular Killing fields $\Omega_{i}$.  In particular, all such multipliers considered in the analysis below have this property.

Finally, we note that, as symmetric traceless two-tensors, the $Q^{(+)}_{n}$ satisfy the Poincar\'{e} inequality\begin{equation}\label{Poincare}
\int_{S^2(r)} |\tilde{\slashed\nabla} Q^{(+)}_{n}|^2 \geq \frac{2}{r^2}\int_{S^2(r)}|Q^{(+)}_{n}|^2,
\end{equation}
owing to the spectrum of the associated spherical Laplacian \eqref{twotensorLaplace}.

\subsection{Degenerate $T$-energy}

The Killing multiplier $T = \partial_{t}$ is a natural starting point in our analysis.  Integrating over the spacetime region bounded by time slices $\{ t = \tau'\}$ and $\{ t = \tau \}$, we obtain the expected conservation law
$$ \int_{\{t = \tau \}} J^{T}_{a}[Q^{(+)}_{n}]\eta^{a} = \int_{\{t = \tau' \}} J^{T}_{a}[Q^{(+)}_{n}]\eta^{a},$$
with $\eta^{a}$ being the appropriate unit normal.

Defining the $T$-energy by
\begin{equation}
E^{T}_{Q^{(+)}_{n}}(\tau) := \int_{\{t = \tau\}} J^{T}_{a}[Q^{(+)}_{n}]\eta^{a},
\end{equation}
our conservation law is nothing more than the statement
\begin{equation}\label{eq: ConservationLaw}
E_{Q^{(+)}_{n}}^{T}(\tau) = E_{Q^{(+)}_{n}}^{T}(\tau')
\end{equation}
for all $\tau$ and $\tau'$.

As the potentials $V^{(+)}_{n}$ are positive on the exterior region, the $T$-energy above is non-negative, degenerating at the event horizon.

\subsection{Red Shift Multiplier}\label{RedShift}

The red shift multiplier $N$, introduced by Dafermos and Rodnianski \cite{DR}, provides a non-degenerate energy.  We recall the details below.

It's convenient to work with the coordinate $(v,R,\theta,\phi)$, where $R=r$ as a function but we use different notation to indicate their coordinate vector fields are different. Consider a vector field $Y$ defined on $\mathcal{H}^+$ by
\begin{equation}
\begin{split}
Y\Big|_{r=2M}&:=-2\frac{\partial}{\partial R},\\
\nabla_Y Y\Big|_{r=2M}&:=-\sigma(T+Y),
\end{split}
\end{equation}
for some $\sigma>0$ to be determined. We have at that $r=2M$ and under the $(v,R,\theta,\phi)$ coordinate,
\[\nabla^{(a}Y^{b)}=\left[ \begin{array}{clclclc} \frac{\sigma}{8} && -\frac{\sigma}{4} && 0 && 0 \\
                                          -\frac{\sigma}{4} && \frac{1}{2M} && 0 && 0 \\
                                          0 && 0 && -\frac{1}{4M^3} && 0 \\
                                          0 && 0 && 0 && -\frac{1}{4M^3\sin^2\theta}
                   \end{array} \right],\]
                   
                  \[ \nabla_a Y^a=-\sigma-\frac{2}{M}.\]
Therefore on the horizon $\mathcal{H}^+$, we have
\begin{equation}
\begin{split}
&K^{Y}[Q^{(+)}_{n}] \\
&= \nabla^a T_{ab}[Q_n^{(+)}]Y^b+T_{ab}[Q_n^{(+)}]\nabla^a Y^b\\
                   &= \frac{\sigma}{8}|\slashed{\nabla}_v Q_n^{(+)}|^2+\frac{1}{2M}|\slashed{\nabla}_R Q_n^{(+)}|^2+\frac{1}{M}\slashed{\nabla}_v Q_n^{(+)}\cdot \slashed{\nabla}_R Q_n^{(+)}\\
                   &+\frac{\sigma}{2}|\tilde{\slashed{\nabla}}Q^{(+)}_n|^2
                   + \left(-\frac{1}{2}(\nabla_Y V_n^{(+)})+\frac{1}{2}\left(\sigma+\frac{2}{M}\right)V_n^{(+)}\right)|Q_n^{(+)}|^2
\end{split}
\end{equation}

Save for the lowest harmonic $n = 2$, the potentials $V_{n}^{(+)}$ are radially increasing functions near the event horizon, and give a positive first term in the $Y$-density above.  We are able to choose $\sigma$ large, to account also for the lowest harmonic $n = 2$, such that the estimate
\begin{equation}
K^{T + Y}[Q^{(+)}_{n}] = K^{Y}[Q^{(+)}_{n}] \geq c T_{ab}[Q^{(+)}_{n}](T+Y)^{a}(T+Y)^{b},
\end{equation}
holds on the horizon and for $c$ a uniform positive constant.

By extending $Y$ smoothly to the exterior region such that $Y$ is non-spacelike and $Y = 0$ as $r \geq R_1$, we obtain 
\begin{align*}
K^{N}[Q^{(+)}_{n}] &\geq{c J^{N}_{a}[Q^{(+)}_{n}]N^{a}}\ &&\textup{for}\ 2M\leq r \leq r_1,\\
J_{a}^{N}[Q^{(+)}_{n}]T^{a} &\sim J_{a}^{T}[Q^{(+)}_{n}]T^{a} \ &&\textup{for}\ r_1 \leq r \leq R_1,\\
|K^{N}[Q^{(+)}_{n}]| &\leq  C|J_{a}^{T}[Q^{(+)}_{n}]T^{a}| \ &&\textup{for}\ r_1 \leq r \leq R_1,\\
N &= T \ &&\textup{for}\ r\geq R_1,
\end{align*}
for some fixed $r_1\in (2M,R_1)$ from a continuity argument. More concretely, we can define $Y$ by first specifying
\begin{equation}\label{hatY}
\hat{Y} := \frac{1}{1-\mu}\partial_{u} = \frac{1}{1-\mu}(\partial_{t} - \partial_{r_{*}}).
\end{equation}

The extension of $Y$ takes the form 
\begin{equation}
Y = f_1(r)\hat{Y} + f_2(r)T,
\end{equation}
where $f_1(r)$ and $f_2(r)$ are non-negative functions with $f_1(2M)=2,\ f_2(2M)=0,\ f'_1(2M)=\sigma,\ f'_2(2M)=\frac{\sigma}{2}$ and $f_1(r)=f_2(r)=0$ for $r\geq R_1$.
\subsection{The Morawetz Multiplier $X$}
Following Morawetz \cite{Morawetz2}, we let $X = f(r)\partial_{r_{*}},$ with $f$ a general radial function, and denote by $\omega^{X}$ a general weight function.  Combining the two, we have the weighted energy current
\begin{equation}
 J^{X,\omega^{X}}_{a}[Q^{(+)}_{n}] := J^{X}_{a}[Q^{(+)}_{n}] + \frac{1}{4}\omega^{X}\nabla_{a}|Q^{(+)}_{n}|^2 -\frac{1}{4}\nabla_{a}\omega^{X}|Q^{(+)}_{n}|^2, 
 \end{equation}
and the weighted density
\begin{equation}
K^{X,\omega^{X}}[Q^{(+)}_{n}] := K^{X}[Q^{(+)}_{n}] + \frac{1}{4}\omega^{X}\Box|Q^{(+)}_{n}|^2 -\frac{1}{4} \Box \omega^{X} |Q^{(+)}_{n}|^2.
\end{equation}

Using the notation $(\,\,)'$ to denote differentiation by the Regge-Wheeler coordinate $r_{*}$, we calculate the unweighted density to be
\begin{align}
\begin{split}
K^{X}[Q^{(+)}_{n}] &= f'|\slashed\nabla_{r_{*}}Q^{(+)}_{n}|^2 + \frac{f}{r}\left(1-\frac{3M}{r}\right)|\tilde{\slashed\nabla}Q^{(+)}_{n}|^2\\
&+\left[\frac{1}{2}(\partial_{r}\mu) fV^{(+)}_{n} -\frac{1}{2}fV^{(+)'}_{n}\right]|Q^{(+)}_{n}|^2\\
&-\frac{1}{4}\left(f' + 2f\frac{1-\mu}{r}\right)\Box |Q^{(+)}_{n}|^2,
\end{split}
\end{align}
where we have used the identity
\begin{equation}\label{productRule}
\Box |Q^{(+)}_{n}|^2 = 2V^{(+)}_{n}|Q^{(+)}_{n}|^2 + 2\slashed{\nabla}^{c}Q^{(+)}_{n}\cdot\slashed\nabla_{c}Q^{(+)}_{n}.
\end{equation}

Inserting the weight function
\begin{equation}
\omega^{X} := f' + 2f\frac{1-\mu}{r},
\end{equation}
we calculate the weighted density
\begin{align}\label{weightedDensity}
\begin{split}
K^{X,\omega^{X}}[Q^{(+)}_{n}] = &f'|\slashed\nabla_{r_{*}}Q^{(+)}_{n}|^2 + \frac{f}{r}\left(1 - \frac{3M}{r}\right)|\tilde{\slashed\nabla} Q^{(+)}_{n}|^2\\
& + \left[\frac{1}{2}(\partial_{r}\mu) fV^{(+)}_{n}-\frac{1}{2}fV^{(+)'}_{n} - \frac{1}{4}\Box\omega^{X}\right]|Q^{(+)}_{n}|^2.
\end{split}
\end{align}

As in the earlier co-closed analysis of \cite{HK}, we utilize the Holzegel multiplier, with
\begin{equation}
f := \left(1-\frac{3M}{r}\right)\left(1 + \frac{M}{r}\right)^2.
\end{equation}

We calculate 
\begin{equation}
f' = \frac{M}{r^2}\left(1+\frac{M}{r}\right)\left(1+\frac{9M}{r}\right)(1-\mu),
\end{equation}
so that the radial and angular terms in \eqref{weightedDensity} above have non-negative sign.  The more complicated portion amounts to
\begin{align}
\begin{split}
&\frac{1}{2 r^8 (3 M + n r)^3}\Big(-4050 M^8 - 162 M^7 (-24 + 29 n) r \\
& -54 M^6 (-24 - 56 n + 37 n^2) r^2  -18 M^5 (51 - 80 n - 26 n^2 + 19 n^3) r^3  \\
  &-M^4 (243 + 630 n - 636 n^2 + 68 n^3) r^4  \\
  &+M^3 (108 - 315 n - 6 n^2 + 152 n^3) r^5  \\
  &+M^2 n (108 - 165 n + 46 n^2) r^6 + M (36 - 49 n) n^2 r^7  +8 n^3 r^8\Big),
\end{split}
\end{align}
a positive quantity save for a spatially compact region away from the photon sphere $r = 3M$.

Integrating over the spherically symmetric spacetime region bounded by time slices $\{ t = \tau'\}$ and $\{ t = \tau \}$ and borrowing from the angular term via the Poincar\'{e} inequality \eqref{Poincare}, we have a base term
\begin{align}
\begin{split}
&\frac{1}{2 r^8 (3 M + n r)^3}\Big(-4050 M^8 - 162 M^7 (-30 + 29 n) r  \\
 &-54 M^6 (-48 - 74 n + 37 n^2) r^2 -18 M^5 (63 - 152 n - 44 n^2 + 19 n^3) r^3  \\
 &-M^4 (675 + 846 n - 1068 n^2 + 32 n^3) r^4 \\
 &+M^3 (216 - 747 n - 78 n^2 + 200 n^3) r^5  \\
 &+M^2 n (216 - 309 n + 38 n^2) r^6 + M (72 - 65 n) n^2 r^7  +12 n^3 r^8\Big),
\end{split}
\end{align}
positive for each harmonic.  Indeed, we obtain a uniform estimate
\begin{equation}
\int_{\{\tau' \leq t \leq \tau\}} \frac{1}{r^3}|Q^{(+)}_{n}|^2 + \frac{1}{r^2}|\slashed\nabla_{r_{*}}Q^{(+)}_{n}|^2 \leq C\int_{\{\tau' \leq t \leq \tau\}} K^{X,\omega^{X}}[Q^{(+)}_{n}],
\end{equation}
with $C$ a constant independent of the harmonic number.

To complete the estimate, we control the boundary terms with our degenerate $T$-energy $E^{T}_{Q^{(+)}_{n}}(\tau)$.  It suffices to estimate each of the pieces of the weighted energy current
 $$ J^{X,\omega^{X}}_{a}[Q^{(+)}_{n}] := J^{X}_{a}[Q^{(+)}_{n}] + \frac{1}{4}\omega^{X}\nabla_{a}|Q^{(+)}_{n}|^2 -\frac{1}{4}\nabla_{a}\omega^{X}|Q^{(+)}_{n}|^2.$$
 
Noting that $\eta^{a} = (1-\mu)^{-1/2}\partial_{t}$, we compute
 $$J^{X}_{a}[Q^{(+)}_{n}]\eta^{a} = f(1-\mu)^{-1/2}\slashed\nabla_{t}Q^{(+)}_{n}\cdot\slashed\nabla_{r_{*}}Q^{(+)}_{n} = f(1-\mu)^{1/2}\slashed\nabla_{t}Q^{(+)}_{n}\cdot\slashed\nabla_{r}Q^{(+)}_{n},$$
 such that
 $$ \Big{|}\int_{\{t = \tau\}} J^{X}_{a}[Q^{(+)}_{n}]\eta^{a} \Big{|} \leq C\int_{\{t = \tau \}} \left[|\slashed\nabla_{t}Q^{(+)}_{n}|^2 + (1-\mu)|\slashed\nabla_{r}Q^{(+)}_{n}|^2\right] \leq CE^{T}_{Q^{(+)}_{n}}(\tau),$$
where we have used uniform boundedness of $f$.  

For the second term, direct computation gives
 $$ \frac{1}{4}\omega^{X}\nabla_{a}|Q^{(+)}_{n}|^2 \eta^{a} = \frac{2r^2-3Mr+3M^2}{2r^3}\left(1+\frac{M}{r}\right)(1-\mu)^{1/2}\slashed\nabla_{t}Q^{(+)}_{n}\cdot Q^{(+)}_{n}.$$
Applying Young's inequality, we deduce 
$$ \Big{|}\int_{\{t = \tau\}} \frac{1}{4}\omega^{X}\nabla_{a}|Q^{(+)}_{n}|^2 \eta^{a} \Big{|} \leq CE^{T}_{Q^{(+)}_{n}}(\tau),$$
with $C$ a universal constant.

As $\omega^{X}$ is purely radial, the third term in the weighted density vanishes upon contraction.  Putting it all together, we bound the weighted energy flux by the initial $T$-energy; that is,
 \begin{equation}
 \Big{|}\int_{\{t = \tau\}} J^{X,\omega^{X}}_{a}[Q^{(+)}_{n}]\eta^{a}\Big{|} \leq C E^{T}_{Q^{(+)}_{n}}(\tau) = CE^{T}_{Q^{(+)}_{n}}(\tau'),
 \end{equation}
 utilizing as well conservation of the $T$-energy.
  
Applying Stokes' theorem, we obtain the following integrated decay estimate
\begin{align}
\begin{split}
&\int_{\{\tau' \leq t \leq \tau\}}\left[\frac{1}{r^2}|\slashed\nabla_{r_{*}}Q^{(+)}_{n}|^2 + \frac{1}{r^3}|Q^{(+)}_{n}|^2\right]\\
&\leq{C\int_{\{\tau' \leq t \leq \tau\}} K^{X,\omega^{X}}[Q^{(+)}_{n}]} \leq {C\int_{\{ t = \tau'\}} J^{T}_{a}[Q^{(+)}_{n}]\eta^{a}} = CE^{T}_{Q^{(+)}_{n}}(\tau'),
\end{split}
\end{align}
with $C$ a universal constant.

\subsection{The Quasi-conformal Multiplier $Z$}

The multiplier we make use of is the analog of the Minkowskian conformal Killing field $Z$ \cite{Morawetz1}, defined to be
\begin{equation}
Z := u^2\partial_{u} + v^2\partial_{v} = \frac{1}{2}(t^2+r_{*}^2)\partial_{t} + tr_{*}\partial_{r_{*}},
\end{equation}
in either the Eddington-Finkelstein or Regge-Wheeler coordinates.  Note that we use the coordinate normalization \eqref{rStar} for $r_{*}$.

We define the $Z$-energy on time-slices $\{ t = \tau \}$ by
\begin{equation}\label{Zenergy}
E^{Z}_{Q^{(+)}_{n}}(\tau) := \int_{\{t = \tau\}} J^{Z}_{a}[Q^{(+)}_{n}]\eta^{a}.
\end{equation}

Using the identity \eqref{productRule}, we compute
\begin{align}
\begin{split}
K^{Z}[Q^{(+)}_{n}] &= t\left(-1 - \frac{\mu r_{*}}{2r} + \frac{r_{*}(1-\mu)}{r}\right)|\tilde{\slashed\nabla}Q^{(+)}_{n}|^2 \\
&- \frac{tr_{*}}{2r}(1-\mu)\Box|Q^{(+)}_{n}|^2 +\Big(\frac{tr_{*}}{r}(1-\mu)V^{(+)}_{n}- \frac{1}{2}V^{(+)}_{n}\nabla_{a}Z^{a}\\
& - \frac{1}{2}Z^{a}\nabla_{a}V^{(+)}_{n}\Big)|Q^{(+)}_{n}|^2,
\end{split}
\end{align}
with the explicit calculations
\[ \nabla_{a}Z^{a} = t(2+\frac{r_{*}}{r}(2-\mu)),\]
\[ Z^{a}\nabla_{a} V^{(+)}_{n} = tr_{*}\partial_{r_{*}}V^{(+)}_{n}.\]

As with the $X$ multiplier, we apply integration by parts, encoded by a weight function $\omega^{Z}$, to swap the d'Alembertian.  With weight function
\begin{equation}
\omega^{Z} = \frac{2tr_{*}}{r}(1-\mu),
\end{equation}
we calculate the weighted density
\begin{align}
\begin{split}
K^{Z,\omega^{Z}}[Q^{(+)}_{n}] &= t\left(-1 - \frac{\mu r_{*}}{2r} + \frac{r_{*}(1-\mu)}{r}\right)|\tilde{\slashed\nabla}Q^{(+)}_{n}|^2 \\
&- \frac{1}{4}\Box\omega^{Z}|Q^{(+)}_{n}|^2+t\Big(\frac{r_{*}}{r}(1-\mu)V^{(+)}_{n}\\
&-\frac{1}{2}V^{(+)}_{n}(2+\frac{r_{*}}{r}(2-\mu))-\frac{1}{2}r_{*}\partial_{r_{*}}V^{(+)}_{n}\Big)|Q^{(+)}_{n}|^2
\end{split}
\end{align}

The first two terms have coefficients independent of the harmonic number, and are well known to be positive near the event horizon and near infinity.  For the third term, we calculate
\begin{align}
\begin{split}
&\frac{r_{*}}{r}(1-\mu)V^{(+)}_{n} -\frac{1}{2}V^{(+)}_{n}(2+\frac{r_{*}}{r}(2-\mu))-\frac{1}{2}r_{*}\partial_{r_{*}}V^{(+)}_{n} \\
&=\frac{2M}{r^4(nr+3M)^3}\Bigg(
 243 M^4 - 162 M^4 n - 162 M^3 r + 351 M^3 n r\\
 & - 162 M^3 n^2 r - 162 M^2 n r^2 + 198 M^2 n^2 r^2 - 48 M^2 n^3 r^2\\
 &- 57 M n^2 r^3 + 52 M n^3 r^3 - 3 n^2 r^4 -14 n^3 r^4\\
 &+ 2\ln(\frac{r-2M}{M})\Big(27 M^4 (-3 + 2 n) + 9 M^3 (3 - 11 n + 6 n^2) r \\
&+ M^2 n (27 - 48 n + 16 n^2) r^2 + 3 M (3 - 4 n) n^2 r^3 + 2 n^3 r^4\Big)\Bigg).
\end{split}
\end{align}
Note that the behavior of this quantity near the event horizon and near infinity is dominated by the logarithmic term, yielding positivity of the quantity in both regimes.  We are able to choose radii $2M < r_2 < R_2 < \infty$ yielding the uniform control
\begin{align}
\begin{split}
K^{Z,\omega^Z}[Q^{(+)}_{n}] &\geq 0\ \textup{as}\ r\leq r_2\ \textup{or}\ r\geq R_2,\\
\Big{|}K^{Z,\omega^Z}[Q^{(+)}_{n}]\Big{|}&\leq Ct\left(|Q^{(+)}_{n}|^2+|\slashed\nabla Q^{(+)}_{n}|^2\right)\ \textup{as}\ r_2\leq r\leq R_2.
\end{split}
\end{align}

With this and the integrated decay estimate for $X$ in hand, we utilize the bootstrap scheme of Dafermos and Rodnianski \cite{DR} to obtain decay estimates for the $Q^{(+)}_{n}$, collected below.

\subsection{Estimates for $Q^{(+)}_{n}$}

In this subsection we collect various estimates following from the vector field multiplier analysis above.  Although there are a variety of boundedness and decay estimates, we present only those relevant for our purposes.

First, we remind the reader of the relevant notation.  Decay is expressed with respect to the foliation of smooth spacelike hypersurfaces $\tilde{\Sigma}_{\tau}$ \eqref{decayFoliation}, characterized by
\begin{align*}
\tau &= t+2M\log(r-2M)+c_0\textup{, for}\ r\leq 3M,\\
 &= t-\sqrt{r^2+1}+c_1\textup{, for}\ r \geq 20M,
\end{align*}
with the specification in the spatially precompact region $3M < r < 20M$ and the choice of constants $c_0$ and $c_1$ made in such a way that $u,v\geq \tau$ on $\tilde{\Sigma}_\tau$.  

Further, we consider the $N$-energy on $\tilde{\Sigma}_{\tau}$ \eqref{NenergyXi}
\[E^{N}_{Q^{(+)}_{n}}(\tilde{\Sigma}_{\tau}) := \int_{\tilde{\Sigma}_{\tau}} J^{N}_{a}[Q^{(+)}_{n}]\eta^{a}_{\tilde{\Sigma}_{\tau}}\]
and the initial energies (\ref{E0def},\ \ref{E1def},\ \ref{E2def})
\begin{align*}
E_0[Q^{(+)}_{n}] &:= \sum_{(q)\leq 2}\int_{\{t = 0 \}} J^{N}_{a}[\mathcal{K}^{(q)}Q^{(+)}_{n}]\eta^{a},\\
E_1[Q^{(+)}_{n}] &:= \sum_{(q)\leq 4} \int_{\{t = 0 \}} (1+r_*^2)J^{N}_{a}[\mathcal{K}^{(q)}Q^{(+)}_{n}]\eta^{a},\\
E_2[Q^{(+)}_{n}] &:= \sum_{(q)\leq 6} \int_{\{t = 0 \}} (1+r_*^2)J^{N}_{a}[\mathcal{K}^{(q)}Q^{(+)}_{n}]\eta^{a},
\end{align*}
where the sum is taken over multi-indices $(q)$ of length $q$ and less, over all Killing commutators $\mathcal{K}$.

\begin{theorem}
Suppose $Q^{(+)}_{n}$ is a solution of \eqref{RW3}, smooth and compactly supported at $\{ t = 0 \}$.  Further, assume that $Q^{(+)}_{n}$ is supported on the harmonic associated with $n$ \eqref{nDefinition}.  Then the $Z$-energy of $Q^{(+)}_{n}$ \eqref{Zenergy} satisfies the uniform bound
\begin{equation}\label{alphaBound}
E^{Z}_{Q^{(+)}_{n}}(\tau) \leq CE_0[Q^{(+)}_{n}].
\end{equation}
Moreover, $Q^{(+)}_{n}$ satisfies the uniform $N$-energy estimate
\begin{equation}\label{alphaNdecay}
E^{N}_{Q^{(+)}_{n}}(\tilde{\Sigma}_{\tau}) \leq CE_1[Q^{(+)}_{n}]\tau^{-2},
\end{equation}
and the uniform decay estimate
\begin{equation}\label{alphaDecay}
\sup_{\tilde{\Sigma}_{\tau}} |Q^{(+)}_{n}| \leq C\sqrt{E_2[Q^{(+)}_{n}]}\tau^{-1}.
\end{equation}
on the family of hypersurfaces $\tilde{\Sigma}_{\tau}$ specified.  Note that the constant $C$ can be regarded as universal, i.e. independent of harmonic number.
\end{theorem}

In the interests of simplifying the metric-level decay estimates to come, we find it convenient to define the symmetric traceless two-tensor $Q^{(+)}$ by specifying $Q^{(+)} := Q^{(+)}_{nm}$ at the harmonic pair $(n,m)$.  Summing the energy estimates above in $L^2(S^2)$, we have associated estimates for $Q^{(+)}$. 

\section{Analysis of the Closed Solution}

With the estimates in place for the gauge-invariant Zerilli-Moncrief function $Z^{(+)}_{n}$, equivalently $Q^{(+)}_{n}$, we complete the analysis of the closed solution by imposing the Chandrasekhar gauge, in which we prove decay estimates on the linearized metric coefficients of the closed solution.  The analysis amounts to rewriting of these coefficients in terms of the $Z_{n}^{(+)}$ and applying the decay estimates of the previous section.  At the outset, we assume that our closed solution $h_1$ is smooth and compactly supported away from the bifurcation sphere on the $\{ t = 0 \}$ time-slice.

\subsection{The Chandrasekhar Gauge}\label{Chandra_gauge}
We begin by defining a constraint operator $\mathcal{L}$ that is related to this gauge condition.
\begin{definition}
Suppose a symmetric two-tensor $h_1$ on Schwarzschild is of the closed form \eqref{h_1}.
The constraint operator $\mathcal{L}$ on such a symmetric two-tensor is defined by 
\begin{equation}\label{L_operator} \mathcal{L}(h_1)=  \frac{1}{2} \mathring{\Delta} H_2+H_2 -H- (r^2-2Mr) h_{11}.\end{equation}
\end{definition}
\begin{lemma}\label{ChandraGauge} For any $h^*$ of the closed form \eqref{h_1}, there is a co-vector $X$ of the form \eqref{G} such that $h=h^*-\pi_{X}$ is still of  closed form \eqref{h_1} and satisfies
 \[h_{01}=0 = H_{A},\] 
 and the initial gauge condition  $\mathcal{L}(h)=0$ at $t=0$.  
 
Fixing radii $2M < r_0 < R_0 < \infty$, if $h^{*}$ is supported away from the bifurcation sphere at $\{ t = 0\}$, say in the radial region $r_0 < r < R_0$, there exist co-vector fields $X_{\RN{1}}$ and $X_{\RN{2}}$ giving the same reduction and initial condition above, with $h^{*} - \pi_{\RN{1}}$ supported away from the bifurcation sphere in the interval $r_0 < r < \infty$ and $h^{*} - \pi_{\RN{2}}$ supported away from spatial infinity in the interval $2M < r < R_0$ at $t = 0$.
\end{lemma}
\begin{proof} Let $X=G_A dx^A+(\mathring{\nabla}_\alpha G_2) dx^\alpha $.  It follows from \eqref{pi_1} that $\pi_{X}$ and $h=h^*- \pi_X$ are both of closed form. We deal with the equation $H_{A}=0$ first. From \eqref{G} we deduce that $G_A$ is determined by 
$\tilde{\nabla}_A G_2-2 (r^{-1} \partial_A r) G_2+G_A=H^*_A$, or
\begin{equation}\label{G_A}
\begin{split}
G_0&=H^*_0-\partial_t G_2 \\
G_1&=H^*_1-\partial_r G_2+2 r^{-1} G_2.
\end{split}
\end{equation}

The equation $h_{01}=0$ is equivalent to 
\[\partial_t G_1+\partial_r G_0-2\Gamma_{10}^0 G_0=h^*_{01}.\]

Plugging the expressions for $G_A$ from \eqref{G_A}, we deduce that 
\[2 \partial_t (-\partial_r G_2+ r^{-1} G_2+\Gamma_{10}^0 G_2)=h^*_{01}-\partial_t H^*_1-\partial_r H^*_0+2\Gamma_{10}^0 H^*_0,\] or
\begin{equation}\label{hyper_eq}
\begin{split}
&-2 \partial_t [r^{1/2} (r-2M)^{1/2} \partial_r (r^{-1/2} (r-2M)^{-1/2} G_2)]\\
&=h^*_{01}-\partial_t H^*_1-\partial_r H^*_0+2\Gamma_{10}^0 H^*_0.
\end{split}
\end{equation}
This is a hyperbolic equation for $G_2$ which can be solved by integrating, subject to suitable initial and boundary conditions. 

Rewriting the last component of $\pi_{X}$, we derive
\[\begin{split} \mathcal{L}(\pi_{X})&=2G_2-2 r( \partial^A r) G_A-2 (r^2-2Mr)(\partial_r G_1-\Gamma_{11}^1 G_1)\\
&=2G_2-2(r-M)G_1-2(r^2-2Mr)\partial_r G_1.\end{split}\]

Plugging in \eqref{G_A}, we obtain:
\[\begin{split}\mathcal{L}(\pi_{X})&=2\left(r^2-2Mr\right)\partial_r^2 G_2-{2}\left(r-{3M}\right)\partial_r G_2+{2}\left(1-\frac{2M}{r}\right) G_2\\
&-2(r-M) H^*_1-2 (r^2-2Mr)\partial_r H^*_1\\
&=2 r^2 \partial_r [r^{-1/2} (r-2M)^{3/2} \partial_r (r^{-1/2} (r-2M)^{-1/2} G_2)]\\
&   -2(r-M) H^*_1-2 (r^2-2Mr)\partial_r H^*_1.      \end{split}\]

It suffices to solve the equation \[\mathcal{L}(\pi_{X})=\mathcal{L}(h^*)=\frac{1}{2} \mathring{\Delta} H^*_2+H^*_2 -H^*- (r^2-2Mr) h^*_{11}, \] 
or 
\begin{equation}\label{initial_eq}
\begin{split}
&2r^2 \partial_r [r^{-1/2} (r-2M)^{3/2} \partial_r (r^{-1/2} (r-2M)^{-1/2} G_2)]\\
&=2(r-M) H^*_1+2 (r^2-2Mr)\partial_r H^*_1+ \frac{1}{2} \mathring{\Delta} H^*_2\\
&+H^*_2 -H^*- (r^2-2Mr) h^*_{11} 
\end{split}    
\end{equation} 
on the $t=0$ slice, regarded as a second order inhomogeneous ODE for $G_2$.  Integrating the equation from either the event horizon or spatial infinity, we can determine mutually exclusive solutions corresponding to the co-vector fields $X_{\RN{1}}$ and $X_{\RN{2}}$, with support in large radii ( $r_0 < r < \infty$) and small radii ($2M < r < R_0$) at $t = 0$, respectively.

The ambiguity of $G_2$, and hence of the Chandrasekhar gauge, consists of solutions of the hyperbolic equation
\begin{equation}\label{hyperbolicEquation}
\partial_t [r^{1/2} (r-2M)^{1/2} \partial_r (r^{-1/2} (r-2M)^{-1/2} G_2)]=0,
\end{equation}
with the initial condition 
 \begin{equation}\label{initialCondition}
 \partial_r [r^{-1/2} (r-2M)^{3/2} \partial_r (r^{-1/2} (r-2M)^{-1/2} G_2)]=0
 \end{equation}
 at $t = 0$.  The general solution of this second order ODE is of the explicit form:
   \begin{equation}\label{residual} c_1 r^{1/2} (r-2M)^{1/2} \int_A^r s^{1/2} (s-2M)^{-3/2} ds+ c_2 r^{1/2} (r-2M)^{1/2},\end{equation} for  $c_1=c_1(\theta, \phi), c_2=c_2(t,\theta, \phi) $ and constant $A>2 M$.
By choosing particular solutions corresponding to either of the co-vectors $X_{\RN{1}}$ or $X_{\RN{2}}$, we can resolve this ambiguity by taking the zero solution to the homogeneous equation \eqref{initialCondition}.  We define the associated normalized solutions by
\begin{align}
&h_{\RN{1}} := h^{*} - \pi_{\RN{1}}\label{hI},\\
&h_{\RN{2}} := h^{*} - \pi_{\RN{2}}\label{hII}.
\end{align}
\end{proof}

\begin{definition}\label{chandra_gauge}
A symmetric two-tensor $h_1$ on Schwarzschild is said to be in the Chandrasekhar gauge if $h_1$ takes the form
\begin{equation}\label{pre_polar_h1} h_1=h_{00} (dt)^2+ h_{11} (dr)^2+(H\mathring{\sigma}_{\alpha\beta}+\mathring{\nabla}_{\alpha}\mathring{\nabla}_\beta {H}_2-\frac{1}{2}\mathring{\sigma}_{\alpha\beta} \mathring{\Delta} {H}_2) dx^\alpha dx^\beta\end{equation} and  $\frac{1}{2} \mathring{\Delta} H_2+H_2 -H- (r^2-2Mr) h_{11}=0$ at $t=0$.
\end{definition}

The previous lemma tells us that we can always reduce a closed solution to Chandrasekhar gauge.  Of particular interest are the normalizations $h_{\RN{1}}$ and $h_{\RN{2}}$, with support away from the bifurcation sphere and away from spatial infinity at $t = 0$, respectively.

\subsection{The Linearized Einstein Equations in Chandrasekhar Gauge}
Subsequently, we will assume a Chandrasekhar gauge has been imposed, with smooth linearized metric coefficients.  Following Chandrasekhar's account \cite{Chandra1}, we use a modified Friedman substitution, defining
\begin{align}\label{friedman}
\begin{split}
N &:= -\frac{1}{2}\left(1-\frac{2M}{r}\right)^{-1}h_{00},\\
L &:= \frac{1}{2}\left(1-\frac{2M}{r}\right)h_{11},\\
T&:= \frac{1}{2r^2}H,\\
V&:= \frac{1}{2r^2}H_2.
\end{split}
\end{align}
Decomposing into spherical harmonics \eqref{projection1}, the closed solution takes the form
\begin{align}
\begin{split}
N^{\ell m} &:= -\frac{1}{2}\left(1-\frac{2M}{r}\right)^{-1}\tilde{h}^{\ell m}_{00},\\
L^{\ell m} &:= \frac{1}{2}\left(1-\frac{2M}{r}\right)\tilde{h}^{\ell m}_{11},\\
T^{\ell m}&:= \frac{1}{2r^2}\tilde{H}^{\ell m},\\
V^{\ell m}&:=\frac{1}{2r^2}\tilde{H}_2^{\ell m},
\end{split}
\end{align}
defined as functions on the quotient space $\mathcal{Q}$.

With our gauge fixing and harmonic decomposition, the behavior of the closed solution is contained in the four quotient functions $N^{\ell m}, L^{\ell m}, T^{\ell m},$ and $V^{\ell m}$.  Subsequently, we will suppress the harmonic dependence in the four functions, using simply $N, L, T,$ and $V$.

Rewriting the Ricci equations in terms of Friedman substitution \eqref{friedman}, 
$\delta R_{0\alpha}=0:$
\begin{equation}\label{talpha}
L_t+\frac{1}{2}(\ell+2)(\ell-1)V_t+T_t=0,
\end{equation}
$\delta R_{01}=0:$
\begin{equation}\label{tr}
-\frac{2L_t}{r}+\frac{2(r-3M)}{r(r-2M)}T_t+2T_{tr}=0,
\end{equation}
$\delta R_{1\alpha}=0:$
\begin{equation}\label{ralpha}
\frac{r-M}{r(r-2M)}L+\frac{r-3M}{r(r-2M)}N-N_r-\frac{1}{2}(\ell+2)(\ell-1)V_r-T_r=0,
\end{equation}
$\delta G_{11}=0:$
\begin{align}\label{rr}
\begin{split}
-&\frac{2}{r(r-2M)}L-\frac{\ell(\ell+1)}{r(r-2M)}N-\frac{\ell(\ell+1)(\ell+2)(\ell-1)}{2r(r-2M)}V\\
-&\frac{(\ell+2)(\ell-1)}{r(r-2M)}T+\frac{2}{r}N_r+\frac{2(r-M)}{r(r-2M)}T_r-\frac{2r}{(r-2M)^2}T_{tt}=0,
\end{split}
\end{align}
$\delta R_{00}=0:$
\begin{align}\label{redundant}
\begin{split}
-&\frac{\ell(\ell+1)(r-2M)}{r^3}N+\frac{M(r-2M)}{r^3}L_r+\frac{(r-2M)(2r-M)}{r^3}N_r\\
&+\frac{2M(r-2M)}{r^3}T_r +\frac{(r-2M)^2}{r^2}N_{rr}-L_{tt}+2T_{tt}=0,
\end{split}
\end{align}
$\delta R_{\alpha\beta}\mathring\sigma^{\alpha\beta}=0:$
\begin{align}\label{trace}
\begin{split}
&\frac{\ell^2+\ell+4}{2}L+\frac{1}{2}\ell(\ell+1)N+\frac{1}{2}\ell(\ell+1)(\ell+2)(\ell-1)V\\
&+(\ell+2)(\ell-1)T+(r-2M)L_r-(r-2M)N_r-2(2r-3M)T_r\\
&-r^2(r-2M)T_{rr}+\frac{r^3}{r-2M}T_{tt}=0,
\end{split}
\end{align}
$\widehat{\delta R}_{\alpha\beta}=0:$
\begin{align}\label{traceless}
-(L+N)-2(r-M)V_r-r(r-2M)V_{rr}+\frac{r^2}{r-2M}V_{tt}=0.
\end{align}

\begin{proposition} \label{gaugefreedom} Suppose $h$ is a closed solution of the linearized vacuum Einstein equations, smooth and compactly supported away from the bifurcation sphere at $\{ t = 0\}$.  For the normalizations $h_{\RN{1}}$ and $h_{\RN{2}}$ defined above (\ref{hI},\ref{hII}), expressed in the Friedman substitution, the following equations hold:
\begin{align}
&L+\frac{1}{2}(\ell+2)(\ell-1)V+T = 0,\label{polarone}\\
&T_r = -\frac{(\ell+2)(\ell-1)}{2r}V-\frac{2r-5M}{r(r-2M)}T.\label{polartwo}
\end{align}
\end{proposition}
\begin{proof}
The equation \eqref{talpha} and the initial constraint $\mathcal{L}(h) = 0$ at $t = 0$ yield \eqref{polarone}:\begin{equation}\label{algebraicConstraint}
L+\frac{1}{2}(\ell+2)(\ell-1)V+T = 0.
\end{equation}

Integrating \eqref{tr} and substituting via \eqref{algebraicConstraint}, we deduce
\begin{equation}
T_r = -\frac{(\ell+2)(\ell-1)}{2r}V-\frac{2r-5M}{r(r-2M)}T + \frac{1}{2}R_2(r),
\end{equation}
with $R_2$ an as yet unspecified radial function.

Rewriting \eqref{rr} and \eqref{ralpha} respectively, we obtain
\begin{align*}
N_r=&\frac{\ell(\ell+1)}{2(r-2M)}N+\frac{(\ell+2)(\ell-1)(\ell(\ell+1)r-2M)}{4r(r-2M)}V\\
&+\frac{\ell(\ell+1)r^2-2(\ell^2+\ell+3)Mr+10M^2}{2r(r-2M)^2}T+\frac{r^3}{(r-2M)^2}T_{tt}\\
&+\frac{r-M}{2(r-2M)}R_2(r),
\end{align*}
\begin{align*}
(\ell+2)(\ell-1)V_r=&-\frac{(\ell+2)(\ell-1)r+6M}{r(r-2M)}N-\frac{(\ell(\ell+1)(\ell+2)(\ell-1))}{2(r-2M)}V\\
              &+\frac{-(\ell+2)(\ell-1)r^2+2(\ell^2+\ell-3)Mr+6M^2}{r(r-2M)^2}T\\
              &-\frac{2r^3}{(r-2M)^2}T_{tt}-\frac{M}{r-2M}R_2(r).
\end{align*}

Assuming these four equations, each of the last three equations \eqref{redundant}, \eqref{trace}, \eqref{traceless} is equivalent to
\begin{equation}\label{r1r2}
\frac{(r-2M)}{2}	R_{2,r}+\frac{(r-M)}{r}R_2=0.
\end{equation}

Analysis of the radial ODE \eqref{r1r2} and our assumption of support in large radii ($h_{\RN{1}}$) or small radii ($h_{\RN{2}}$) yield $R_2(r) = 0$.
\end{proof}
We have arrived at a situation analogous to that in Chandrasekhar \cite{Chandra1}; namely, we have the two equations \eqref{polarone} and \eqref{polartwo}, together with
\begin{align}
\begin{split}\label{polarthree}
&N_r=\frac{\ell(\ell+1)}{2(r-2M)}N+\frac{(\ell+2)(\ell-1)(\ell(\ell+1)r-2M)}{4r(r-2M)}V\\
&+\frac{\ell(\ell+1)r^2-2(\ell^2+\ell+3)Mr+10M^2}{2r(r-2M)^2}T+\frac{r^3}{(r-2M)^2}T_{tt},
\end{split}
\end{align}
\begin{align}
\begin{split}\label{polarfour}
&(\ell+2)(\ell-1)V_r=-\frac{(\ell+2)(\ell-1)r+6M}{r(r-2M)}N\\
              &+\frac{-(\ell+2)(\ell-1)r^2+2(\ell^2+\ell-3)Mr+6M^2}{r(r-2M)^2}T\\
               &-\frac{2r^3}{(r-2M)^2}T_{tt}-\frac{(\ell(\ell+1)(\ell+2)(\ell-1))}{2(r-2M)}V.
\end{split}
\end{align}

\begin{lemma}
For either of the normalized solutions $h_{\RN{1}}$ and $h_{\RN{2}}$, the Zerilli-Moncrief function \eqref{Zdef} takes the form
\begin{equation}\label{ChandraZ}
Z^{(+)} = rV + \frac{r^2}{nr+3M}T.
\end{equation}
\end{lemma}
\begin{proof}
We recall the definition of $Z^{(+)}$ (omitting the $\ell m$ index)
\begin{equation}
Z^{(+)} := \frac{2r}{\ell(\ell +1)}\left[\tilde{K} + \frac{2}{\Lambda}\left(r^{A}r^{B}\tilde{k}_{AB} -rr^{A}\tilde\nabla_{A}\tilde{K}\right)\right].
\end{equation}
In the Chandrasekhar gauge, $\tilde{K}$ and  $r^{A}r^{B}\tilde{k}_{AB}$ can be expressed  in terms of $L, T, V$ in \eqref{friedman}, in particular:
\[\tilde{K}=2T+\ell(\ell +1) V+2(r-2M) V_r,\] and \[r^{A}r^{B}\tilde{k}_{AB}=2 (1-\frac{2M}{r})L +2(1-\frac{2M}{r})^2[\partial_r(r^2 V_r)+\frac{Mr}{r-2M} V_r].\]

Plugging these in the expression of $Z^{(+)}$, one checks that all $V_{rr}$ and $V_r$ terms are cancelled, while $L$ and $T_r$ can be substituted using \eqref{polarone} and \eqref{polartwo}, respectively. Collecting the coefficients of $V$ and $T$, we arrive at 
\eqref{ChandraZ}. 
\end{proof}

Indeed, the decoupling of $Z^{(+)}$ can be verified directly using the equations above; see Chandrasekhar \cite{Chandra1}.

\subsection{Metric Coefficients in Terms of $Z^{(+)}$}

Briefly, we describe the procedure for extending our decay estimates to the linearized metric coefficients.  At the outset, we assume that our closed solution $h_1$ is smooth and compactly supported away from the bifurcation sphere on the time-slice $\{ t= 0 \}$.  Choosing gauge co-vectors $X_{\RN{1}}$ and $X_{\RN{2}}$ so that the normalized solutions $h_{\RN{1}}$ \eqref{hI} and $h_{\RN{2}}$ \eqref{hII} are supported away from the bifurcation sphere and spatial infinity on the time slice $\{ t = 0 \}$, respectively, we are able to show decay of the resulting closed solutions in certain radial regimes.  Taking an interpolation of the two gauge choices, it remains to consider a compact radial region, wherein we lose our diagonalization but are nontheless able to estimate the solution.

In expressing the metric-level quantities $N, L, T,$ and $V$ in terms of $Z^{(+)}$, we find it useful to introduce the quantities
\begin{equation}\label{phioneDef}
\Phi_{\RN{1}}^{nm}(t,r) := n\sqrt{1-\frac{2M}{r}}\int_{2M}^{r} {\frac{Z_{nm}^{(+)}(t,r')}{\sqrt{1-\frac{2M}{r'}}(nr'+3M)}dr'},
\end{equation}
\begin{equation}\label{phitwoDef}
\Phi_{\RN{2}}^{nm}(t,r) := -n\sqrt{1-\frac{2M}{r}}\int_r^{\infty} {\frac{Z_{nm}^{(+)}(t,r')}{\sqrt{1-\frac{2M}{r'}}(nr'+3M)}dr'}.
\end{equation}

In what follows, we continue our convention of suppressing harmonic dependence.
As mentioned above, we can choose a co-vector field $X_{\RN{1}}$ in such a way that the normalized solution $h_{\RN{1}}$ is supported away from the bifurcation sphere at $t = 0$.  For such a solution, the following pullback argument applies.

Using the expression for $Z^{(+)}$ in the Chandrasekhar gauge \eqref{ChandraZ}, we can rewrite the equation \eqref{polartwo} as an equation in $T$ and $Z^{(+)}$.  Integrating, we find
\begin{equation}\label{Tpullback}
T(t,r) = -\frac{1}{r^2}(nr+3M)\Phi_{\RN{1}}(t,r).
\end{equation}
Substituting for $T$ in the definition of $Z^{(+)}$, we obtain the relation
\begin{equation}\label{Vpullback}
V(t,r) = \frac{1}{r}(Z^{(+)}(t,r) + \Phi_{\RN{1}}(t,r)).
\end{equation}
Substituting these relations for $T$ and $V$, the algebraic equation \eqref{polarone} yields
\begin{equation}\label{Lpullback}
L(t,r) = -\frac{n}{r}Z^{(+)}(t,r) + \frac{3M}{r^2}\Phi_{\RN{1}}(t,r).
\end{equation}
Finally, using \eqref{polarfour} we find
\begin{align}\label{Npullback}
\begin{split}
N(t,r) &= -\frac{nr}{nr+3M}\partial_{r_{*}}Z^{(+)}(t,r) \\
&- \frac{n}{(nr+3M)^2}\left[\frac{6M^2}{r}+3Mn+n(n+1)r\right]Z^{(+)}(t,r)\\ 
&+ \left(M -\frac{M^2}{r-2M}\right)\frac{\Phi_{\RN{1}}(t,r)}{r^2} + \frac{r^2}{r-2M}\partial_{t}^2\Phi_{\RN{1}}(t,r).
\end{split}
\end{align}

We emphasize that the components $N, L, T, V$ above are those associated with the $h_{\RN{1}}$.  The pullback is entirely analogous for components of the normalized solution $h_{\RN{2}}$ generated by $X_{\RN{2}}$, with $\Phi_{\RN{1}}$ replaced by $\Phi_{\RN{2}}$.

\subsection{Decay of the $\Phi$}
Having written each of the linearized metric coefficients $N, L, T$, and $V$ in terms of $Z^{(+)}$ and the $\Phi$, it remains to prove decay of each.  Decay of $Z^{(+)}$ is the subject of the previous section, and, as we shall see, these results also lead to the decay of the $\Phi$.  Note that we suppress the $m$-dependence, consistent with the estimates for the $Q^{(+)}_{n}$.

From the previous analysis \eqref{alphaBound}, we know that the $Z$-energy of the $Q^{(+)}_{n}$ through time-slices is bounded.  Underestimating the $Z$-energy, we find
\begin{align*}
CE_0[Q^{(+)}_{n}] \geq E^{Z}_{Q^{(+)}_{n}}(\tau) &\geq \int_{\{t = \tau\}} \sqrt{1-\frac{2M}{r}}\frac{u^2 + v^2}{r^2}|Q^{(+)}_{n}|^2\\
&= \int_{2M}^{\infty}\int_{S^2}(u^2 + v^2)|Q^{(+)}_{n}|^2 dr d\mathring\sigma\\
&=\int_{2M}^{\infty}\int_{S^2}\frac{u^2+v^2}{r^2}|Z^{(+)}_{n}|^2 |Y_{\alpha\beta}^{\ell m}|^2 dr d\mathring\sigma\\
&=\int_{2M}^{\infty} \frac{u^2 + v^2}{r^2}n^2|Z^{(+)}_{n}|^2 dr\\
&\geq\int_{2M}^{\infty} \frac{u^2 + v^2}{r^2}|Z^{(+)}_{n}|^2 dr.
\end{align*}

From this and the definition of the $\Phi$ (\ref{phioneDef}, \ref{phitwoDef}), we deduce the decay estimate
\begin{align*}
&|\Phi^{n}(t, r)| \\
&\leq C\sqrt{1-\frac{2M}{r}}\int_{2M}^{\infty}\left|\frac{Z_{n}^{(+)}}{r\sqrt{1-\frac{2M}{r}}}\right| dr\\
&\leq C\sqrt{1-\frac{2M}{r}}\left(\int_{2M}^{\infty} \frac{u^2+v^2}{r^2}|Z_{n}^{(+)}|^2 dr\right)^{1/2}\left(\int_{2M}^{\infty}\frac{1}{1-\frac{2M}{r}} \frac{1}{u^2 + v^2} dr\right)^{1/2}\\
&\leq C\sqrt{1-\frac{2M}{r}}\sqrt{E_0[Q^{(+)}_{n}]}\left(\int_{-\infty}^{\infty} \frac{1}{t^2 + r_{*}^2}dr_{*}\right)^{1/2}
\end{align*}
such that
\begin{equation}\label{phiDecay}
|\Phi^{n}(t,r)|\leq C\sqrt{1-\frac{2M}{r}}\sqrt{E_0[Q^{(+)}_{n}]}t^{-1/2},
\end{equation}
decaying through the \emph{time slices}.  Note that this implies decay through the foliation \eqref{decayFoliation} above.

\subsection{Decay of the Metric Coefficients}

We illustrate decay of the metric coefficients by considering $L$; the arguments for $T$ and $V$ are entirely analogous.

Integrating on the unit sphere, we find
\begin{align}\label{Lprelim}
\begin{split}
\int_{S^2}|L^{\ell m}|^2 &\leq 2\left(\int_{S^2}\Big{|}\frac{nZ^{(+)}_{nm}}{r}\Big{|}^2 + \frac{9M^2}{r^4}\int_{S^2} |\Phi_{nm}|^2\right)\\
&\leq \int_{S^2}|Q^{(+)}_{nm}|^2 + \frac{9M^2}{r^4}\int_{S^2}|\Phi_{nm}|^2\\
&\leq C\left(E_2[Q^{(+)}_{nm}]\tau^{-2} + E_0[Q^{(+)}_{nm}]\tau^{-1}\right).
\end{split}
\end{align}
Summing the spherical harmonics, we deduce
\begin{equation}\label{Lsum}
\begin{split}
\int_{S^2}|L|^2 &= \sum_{\ell m} \int_{S^2} |L^{\ell m}|^2\\
&\leq C\sum_{\ell m}\left(E_2[Q^{(+)}_{nm}]\tau^{-2} + E_0[Q^{(+)}_{nm}]\tau^{-1}\right)\\
&\leq CE_2[Q^{(+)}]\tau^{-1}.
\end{split}
\end{equation}

Such an $L^2$-estimate is also possible upon applying the angular Killing operators $\Omega_{i}$; that is, an analogous estimate to \eqref{Lsum} holds for $\Omega_{i}L$, $\Omega_{i}\Omega_{j}L$, etc.  Applying Sobolev embedding on the spheres, we deduce the pointwise decay estimate
\begin{equation}
\sup_{\tilde{\Sigma}_{\tau}} |L| \leq C \sum_{(q)\leq 2} \sqrt{E_2[\Omega^{(q)}Q^{(+)}]}\tau^{-1/2}
\end{equation}
through the standard decay foliation \eqref{decayFoliation}.  Decay of the components $T$ and $V$ is similar.  Overall, we find the pointwise decay estimates
\begin{align}
&\sup_{\tilde{\Sigma}_{\tau}} |L| \leq C \left(\sum_{(q)\leq 2} \sqrt{E_2[\Omega^{(q)}Q^{(+)}]}\right)\tau^{-1/2},\\
&\sup_{\tilde{\Sigma}_{\tau}} |T| \leq C \left(\sum_{(q)\leq 2} \sqrt{E_0[\Omega^{(q)}Q^{(+)}]}\right)\tau^{-1/2},\\
&\sup_{\tilde{\Sigma}_{\tau}} |V| \leq C\left(\sum_{(q)\leq 2} \sqrt{E_2[\Omega^{(q)}Q^{(+)}]}\right)\tau^{-1/2}
\end{align}
through the foliation \eqref{decayFoliation}.  Note that such estimates hold for either of the normalized solutions $h_{\RN{1}}$ or $h_{\RN{2}}$.

The remaining component $N$ is much less straightforward to estimate.  Fixing radii $2M < r_0 < R_0 < \infty$, we are able to recover decay estimates for the normalized solutions in associated radial regimes.

\subsubsection{The region $2M < r < R_0$}

We consider the closed solution $h_{\RN{1}}$ \eqref{hI} in the radial region $2M < r < R_0$.  Rewriting the last term in the expression for $N$ \eqref{Npullback} by bringing the partial derivatives under the integral, applying the Zerilli equation for $Z^{(+)}$ \eqref{ZerilliEqn}, and integrating by parts, we are left with
\begin{align}
\begin{split}
N = &-\frac{n\left(9M^3 + 3Mnr(1+(-1+n)r)\right)}{r(nr+3M)^2}Z^{(+)}\\
&-\frac{n\left(nr^2(1+n-nr)+6M^2(1+2nr)\right)}{r(nr+3M)^2}Z^{(+)}\\
&+ \frac{M}{r^2}\Phi_{\RN{1}} - \frac{M^2}{r^3(r-2M)}\Phi_{\RN{1}}\\
&- \frac{r}{\sqrt{1-\frac{2M}{r}}}\int_{2M}^{r}g(r')\frac{Z^{(+)}}{\sqrt{1-\frac{2M}{r'}}(nr'+3M)}dr',
\end{split}
\end{align}
where
\begin{equation}\label{g_r}
g(r) := \frac{2nr^2 -4M(-1+n)r - 9M^2}{r^4}.
\end{equation}

The expression for $N$ consists of ``good" terms in the radial region with the possible exception of the last two, the coefficients of which blow up at the event horizon.  We can capture the troublesome behavior of these terms near the event horizon in the analysis of the model quantity
\begin{align*}
\frac{1}{1-\frac{2M}{r}}\Phi_{\RN{1}} &= \frac{1}{\sqrt{1-\frac{2M}{r}}}\int_{2M}^{r} \frac{Z^{(+)}}{\sqrt{1-\frac{2M}{r'}}(r'+3M/n)}dr'.
\end{align*} 
To estimate the model quantity, we note that $Z^{(+)}$ satisfies the decay estimate
\[Z^{(+)}_{n} \leq C\sqrt{E_2[Q^{(+)}_{n}]}v_{+}^{-1}\]
with $v_{+} = \max\{1,v\}$, in the region $\{ t \geq 0\} \cap{\{ r\leq R_0\}}$.  For large $v$, this estimate follows from the pointwise decay estimate on $Q_{n}^{(+)}$ \eqref{alphaDecay} and the form of the foliation \eqref{decayFoliation}.  Otherwise, the estimate reduces to a weaker statement of boundedness on the exterior region.

Recalling the form of the Regge-Wheeler coordinate $r_{*}$ \eqref{rStar}, we note that the comparison
\[ce^{r_{*}/2M} \leq 1 - \frac{2M}{r} \leq C e^{r_{*}/2M}\]
holds for small radii.  For a point $p$ on the quotient space, let 
\[ t= t(p), r = r(p), r_{*} = r_{*}(p), \tau = v(p).\]
We estimate $\Phi_{\RN{1}}$ by 
\begin{align*}
&|\Phi_{\RN{1}}(p)| \leq C\sqrt{1-\frac{2M}{r}}\int_{2M}^{r}\frac{Z_{n}^{(+)}}{\sqrt{1-\frac{2M}{r'}}}dr' \\
&= C\sqrt{1-\frac{2M}{r}}\int_{-\infty}^{r_{*}} \sqrt{1-\frac{2M}{r'}}Z_{n}^{(+)}(t,r_{*}')dr_{*}' \\
&\leq Ce^{r_{*}/4M}\int_{-\infty}^{r_{*}} e^{r_{*}'/4M}\sqrt{E_2[Q^{(+)}_{n}]}\max\{1,t+r_{*}'\}^{-1}dr_{*}'\\
&\leq C\sqrt{E_2[Q^{(+)}_{n}]}e^{(r_{*}-t)/4M}\int_{-\infty}^{\tau}e^{s/4M}\max\{1,s\}^{-1}ds,
\end{align*}
where we have performed the change of variable $s = t + r_{*}'$.
Asymptotically in $\tau$, the function
\[\int_{-\infty}^{\tau}e^{s/4M}\max\{1,s\}^{-1}ds\]
is comparable with $e^{\tau/4M}\tau^{-1}$, as can be seen by application of L'Hospital's rule.  Hence for adequately large $\tau$, the estimate above becomes
\[\Phi_{\RN{1}}\leq C\sqrt{E_2 [Q^{(+)}_n]}e^{(r_{*}-t)/4M}e^{\tau}\tau^{-1} \leq C\sqrt{E_2 [Q^{(+)}_n]}\left(1-\frac{2M}{r}\right)\tau^{-1}.\]
In the region $2M\leq r\leq R_0$, the last term can be estimated in the same way since
\begin{align*}
\left|\frac{r}{\sqrt{1-\frac{2M}{r}}}\int_{2M}^{r}g(r')\frac{Z_{n}^{(+)}}{\sqrt{1-\frac{2M}{r'}}(nr'+3M)}dr'\right|\leq \frac{C}{\sqrt{1-\frac{2M}{r}}}\int_{2M}^{r}\frac{Z_{n}^{(+)}}{\sqrt{1-\frac{2M}{r'}}}dr'.
\end{align*}
Summing over the spherical harmonics, we deduce 
\begin{equation}
\sup_{\tilde{\Sigma}_\tau} |N|\leq C\left( \sum_{(q)\leq 2} \sqrt{E_2 [\Omega^{(q)}Q^{(+)}]}\right)\tau^{-1}
\end{equation}
for the normalized solution $h_{\RN{1}}$ in the region $2M\leq r\leq R_0$.  Indeed, the other three components $L, T, V$ also satisfy this improved estimate for small radii.

\subsubsection{The region $r_0 < r < \infty$}
For the large radii, we consider the solution $h_{\RN{2}}$.  The analysis of $N$ is much the same as the previous subsection.  In particular, we can rewrite $N$ as 
\begin{align}
\begin{split}
N = &-\frac{n\left(9M^3 + 3Mnr(1+(-1+n)r)\right)}{r(nr+3M)^2}Z^{(+)}\\
&-\frac{n\left(nr^2(1+n-nr)+6M^2(1+2nr)\right)}{r(nr+3M)^2}Z^{(+)}\\
&+ \frac{M}{r^2}\Phi_{\RN{2}} - \frac{M^2}{r^2(r-2M)}\Phi_{\RN{2}}\\
&+ \frac{r}{\sqrt{1-\frac{2M}{r}}}\int_{r}^{\infty}g(r')\frac{Z^{(+)}}{\sqrt{1-\frac{2M}{r'}}(nr'+3M)}dr',
\end{split}
\end{align}
with $g$ defined in \eqref{g_r}.

Now the only troublesome term is the last, with coefficient growing large as $r$ approaches spatial infinity.  Quadratic decay of $g$ allows us to estimate
\begin{align*}
&\frac{r}{\sqrt{1-\frac{2M}{r}}}\int_{r}^{\infty}g(r')\frac{Z^{(+)}}{\sqrt{1-\frac{2M}{r'}}(nr'+3M)}dr'\\
&\leq Cr\int_{r}^{\infty}\frac{1}{(r')^2}\frac{Z^{(+)}}{\sqrt{1-\frac{2M}{r'}}(r'+3M/n)}dr' \leq \frac{C}{r}\Phi_{\RN{2}}.
\end{align*}

Summing the spherical harmonics, we obtain a pointwise estimate for $N$,
\begin{equation}
\sup_{\tilde{\Sigma}_{\tau}} |N| \leq C \left(\sum_{(q)\leq 2} \sqrt{E_2[\Omega^{(q)}Q^{(+)}]}\right)\tau^{-1/2},
\end{equation}
for the solution $h_{\RN{2}}$ in the radial region $r_0 < r < \infty$.

\subsection{Decay of the Interpolated Solution}
We define $c_{\ell m}$ in terms of the difference between $\Phi^{\ell m}_{\RN{1}}$ and $\Phi^{\ell m}_{\RN{2}}$:
\begin{align}
\begin{split}
\left(1-\frac{2M}{r}\right)^{-1/2}c_{\ell m}(t) &:= \Phi_{\RN{1}}^{\ell m}-\Phi_{\RN{2}}^{\ell m}\\
                            &=\sqrt{1-\frac{2M}{r}}\int_{2M}^\infty \frac{1}{\sqrt{1-\frac{2M}{r'}}}\frac{Z^{(+)}_{nm}}{r'+3M/n}dr'.
\end{split}
\end{align}

In addition, we define the summation
\begin{equation}
c(t,\theta,\phi) := \sum_{\ell, m} c_{\ell m}(t)Y^{\ell m}(\theta,\phi),
\end{equation}
where the sum is understood on $L^2(S^2)$.

From the decay of the $\Phi$ \eqref{phiDecay}, we have
\begin{align*}
&|c_{\ell m}(t)|\leq \frac{C}{n}\sqrt{E_0[Q^{(+)}_{nm}]} t^{-1/2},\\
&|c(t,\theta,\phi)| \leq C\sqrt{E_0[Q^{(+)}]}t^{-1/2}.
\end{align*}

Defining the function $\check{G}(t,r,\theta,\phi)$ by
\begin{equation}
\check{G}(t,r,\theta,\phi):=-c(t,\theta,\phi)\sqrt{r^2-2Mr},
\end{equation}
we construct an associated co-vector field
\begin{align}\label{checkXDef}
\begin{split}
\check{X}&:=-r^2\tilde{\nabla}_A(r^{-2}\check{G}) dx^A+\mathring{\nabla}_\alpha \check{G}dx^\alpha\\
         &=r\left(1-\frac{2M}{r}\right)^{1/2}\partial_{t}c(t,\theta,\phi) dt-r\left(1-\frac{2M}{r}\right)^{1/2} \mathring\nabla_{\alpha}c(t,\theta,\phi)  dx^\alpha\\
         &-\left(1-\frac{3M}{r}\right)\left(1-\frac{2M}{r}\right)^{-1/2}c(t,\theta,\phi)dr.
\end{split}
\end{align}

Denote the deformation tensor $\check{\pi}:=\pi_{\check{X}}$. By direct computation, $\pi_{\check{X}}$ is seen to be in Chandrasekhar gauge, with corresponding Friedman substitution quantities:
\begin{align*}
N[\check{\pi}]&=-\frac{1}{2}\left(1-\frac{2M}{r}\right)^{-1}\check{\pi}_{00}\\
&=-\frac{M}{r^3}\left(1-\frac{2M}{r}\right)^{-1/2}(r-3M)c(t,\theta,\phi)-r\left(1-\frac{2M}{r}\right)^{-1/2}\partial_t^2 c(t,\theta,\phi),\\
L[\check{\pi}]&=\frac{1}{2}\left(1-\frac{2M}{r}\right)\check{\pi}_{11}=-\frac{3M}{r^2}\left(1-\frac{2M}{r}\right)^{1/2}c(t,\theta,\phi),\\
T [\check{\pi}]&=\frac{1}{2r^2} \check{\pi}_{\alpha\beta}\mathring{\sigma}^{\alpha\beta}=\frac{\left(1-\frac{2M}{r}\right)^{1/2}}{r}\left(-\mathring{\Delta}-2+\frac{6M}{r}\right)c(t,\theta,\phi),\\
V_{\alpha\beta}[\check{\pi}]&=\frac{1}{2r^2}\left(\check{\pi}_{\alpha\beta}-r^2T[\check{\pi}]\mathring{\sigma}_{\alpha\beta}\right)=\frac{\left(1-\frac{2M}{r}\right)^{1/2}}{r}\sum_{\ell m}c_{\ell m}(t)Y^{\ell m}_{\alpha\beta}.
\end{align*}
Comparing the components of $\pi_{\RN{1}},\pi_{\RN{2}}$ and $\check{\pi}$, we observe that $\check{\pi}=\pi_{\RN{2}}-\pi_{\RN{1}}$ and $\check{X}=X_{\RN{2}}-X_{\RN{1}}$. Here we are just using the definition of $c_{\ell m}(t)$ and the expressions for the metric coefficients in terms of $Z^{(+)}$ (\ref{Tpullback}, \ref{Vpullback}, \ref{Lpullback}, \ref{Npullback}).  For any $r_0\leq r\leq R_0$, we have

\begin{align*}
\int_{S^2(t,r)} L[\check{\pi}]^2 d\mathring{\sigma}&\leq C(M,r_0,R_0)\sum_{\ell m}c_{\ell m}^2(t)\\
                                                   &\leq C(M,r_0,R_0)\sum_{\ell m} \frac{1}{n^2} E_0 [Q^{(+)}_{nm}]t^{-2}\\
                                                   &\leq C(M,r_0,R_0) E_0 [Q^{(+)}]t^{-2},
\end{align*}
Similarly,
\begin{align*}
\int_{S^2(t,r)} N[\check{\pi}]^2 d\mathring{\sigma}&\leq C(M,r_0,R_0) E_2[Q^{(+)}]t^{-2},
\end{align*}
\begin{align*}
\int_{S^2(t,r)} T[\check{\pi}]^2 d\mathring{\sigma}&\leq C(M,r_0,R_0)\sum_{\ell m} n^2c_{\ell m}^2(t)\\
                                                   &\leq C(M,r_0,R_0)\sum_{\ell m} E_0[Q^{(+)}_{\ell m}]t^{-2}\\
                                                   &\leq C(M,r_0,R_0) E_0[Q^{(+)}]t^{-2},
\end{align*}
\begin{align*}
\int_{S^2(t,r)}|V_{\alpha\beta}|^2_{\mathring{\sigma}} d\mathring{\sigma} &\leq C(M,r_0,R_0)\sum_{\ell m} n^2c_{\ell m}^2(t)\\
                                                                                         &\leq C(M,r_0,R_0) E_0[Q^{(+)}]t^{-2}.
\end{align*}

Having been defined in terms of the decaying function $\check{G}$, each component of $\check{X}$ has the same type of bound in $r_0\leq r \leq R_0$ on the $L^2$-norm on spheres.  Hence, letting $0\leq \eta(r)\leq 1$ be a cut-off function with $\eta\equiv 0$ as $r\leq r_0$ and $\eta\equiv 1$ as $r\geq R_0$, the deformation tensor
\begin{equation}\label{interpolatedDeformation}
\pi_{\eta(r)\check{X} }=\eta(r)\check{\pi}+\eta'(r)\left(dr\otimes \check{X}+\check{X}\otimes dr \right)
\end{equation}
decays in the same fashion, in this radially compact region.  Defining the interpolated co-vector field
\begin{equation}\label{interpolatedCovector}
X:=X_{\RN{1}}+\eta(r)\check{X}=X_{\RN{2}}-(1-\eta(r))\check{X},
\end{equation}
we note that
\begin{align*}
&h_1-\pi_X= h_{\RN{1}}-\pi_{\eta\check{X}}\ \textup{as}\ r\leq R_0,\\
&h_1-\pi_X= h_{\RN{2}}+\pi_{(1-\eta)\check{X}}\ \textup{as}\ r\geq r_0,
\end{align*}
such that the $L^2$-norm on spheres of each linearized metric component of $h = h_1-\pi_X$ decays through $\tilde{\Sigma}_\tau$.  Commuting with the angular Killing fields $\Omega_{i}$ and applying Sobolev embedding, we obtain further pointwise estimates on the components.  

Away from the interpolation region $r_0 \leq r \leq R_0$ the linearized metric has the Chandrasekhar gauge, and it is convenient to express decay of the metric components in terms of the spacetime norm.  Namely, away from the interpolation region we have
\[ |h|^2_{g} = 4\left(|N|^2 + |L|^2 + |T|^2 + |V|^2\right),\]
with the spacetime norm being positive-definite.  We note that the coordinate frame in question is highly irregular at the event horizon, where $N$ and $L$ coincide; in this way, there is a loss of control of the linearized solution at the event horizon.

Within the interpolation region $r_0 \leq r \leq R_0$, the linearized metric takes the form
\[ h = h_1 - \pi_{X} = h_{\RN{1}} - \eta(r)\check{\pi} - \eta'(r)\left(dr\otimes \check{X}+\check{X}\otimes dr \right).\]
Decay for the components of Chandrasekhar-gauged $h_{\RN{1}}$ and $\check{\pi}$ is captured by the (positive-definite) spacetime norm, as mentioned above.  For the remaining components we utilize the estimates on $c(t,\theta,\phi)$ and $\check{X}$ derived earlier in this subsection to control the squared norms by the decaying quantity
\[ C(r_0,R_0,M)(E_1[Q^{(+)}])\tau^{-1}.\]
In total, we overestimate the squared norms of the linearized metric components by means of
\[ |h_1|^2_{g} + |\check{\pi}|^2_{g} + C(r_0,R_0,M)\left(E_1[Q^{(+)}]\right)\tau^{-1}.\]

\begin{theorem}\label{closedMain}
Suppose $h_1$ is a closed solution of the linearized vacuum Einstein equations \eqref{linearized_Einstein}, with support in $\ell \geq 2$.  Further, assume that $h_1$ is smooth and compactly supported away from the bifurcation sphere on $\{ t = 0\}$.  Normalizing $h_1$ by the co-vector \eqref{interpolatedCovector},
$$ h = h_1 - \pi_{X},$$
such that $h$ is in the interpolated Chandrasekhar gauge, the normalized solution $h$ satisfies the following spacetime decay estimates:

Away from the interpolation region $r_0 \leq r \leq R_0$, we have
\begin{equation}\label{metricDecay2}
\sup_{\tilde{\Sigma}_{\tau}} |h|_{g}\leq C \left(\sum_{(q)\leq 2} \sqrt{E_2[\Omega^{(q)}Q^{(+)}]}\right)\tau^{-1/2}.
\end{equation}
through the decay foliation \eqref{decayFoliation}.  As mentioned above, the spacetime norm is positive-definite and describes the sum-of-squares of the linearized metric components.

Restricting further to the region $2M \leq r \leq r_0$, we have the improvement
\begin{equation}
\sup_{\tilde{\Sigma}_{\tau}} |h|_{g}\leq C(r_0) \left(\sum_{(q)\leq 2} \sqrt{E_2[\Omega^{(q)}Q^{(+)}]}\right)\tau^{-1}.
\end{equation}

Finally, considering the interpolation region $r_0 \leq r \leq R_0$, the norms of the linearized metric components are dominated by the quantity
\begin{equation}
C(r_0,R_0,M)\left(|h_1|_{g} + |\check{\pi}|_{g} + \sqrt{E_1[Q^{(+)}]}\tau^{-1/2}\right),
\end{equation}
itself satisfying the decay estimate
\begin{align}
\begin{split}
&\left(|h_1|_{g} + |\check{\pi}|_{g} + \sqrt{E_1[Q^{(+)}]}\tau^{-1/2}\right) \\
&\leq C(r_0,R_0,M) \left(\sum_{(q)\leq 2} \sqrt{E_2[\Omega^{(q)}Q^{(+)}]}\right)\tau^{-1/2}.
\end{split}
\end{align}
\end{theorem}


\section{Proof of Theorem 2}

In this final section, we combine the results of Sections 6-9, on the analysis of the closed and co-closed portions, to obtain decay estimates for a suitable normalization of $\delta g^{\ell \geq 2}$, corresponding to Theorem 2 in the Introduction.  With the lower modes of $\delta g^{\ell < 2}$ accounted for in Section \ref{lowerModesSection}, the following theorem is a statement of linear stability for the Schwarzschild spacetime.
\begin{theorem}
Suppose $\delta g^{\ell \geq 2}$ is a solution of the linearized vacuum Einstein equations \eqref{linearized_Einstein}, with support in $\ell \geq 2$.  Moreover, assume that $\delta g^{\ell \geq 2}$ is smooth and compactly supported away from the bifurcation sphere on the time-slice $\{ t = 0\}$.  Then there exists a smooth co-vector $X^{\ell \geq 2}$  such that 
\begin{equation}
\delta g^{\ell \geq 2} = \pi_{X^{\ell \geq 2}} +\widehat{\delta g}^{\ell \geq 2},
\end{equation} 
where the norms of the linearized metric components of the normalized solution $\widehat{\delta g}^{\ell \geq 2}$ are bounded by the decaying expression
\begin{equation}\label{metricDecay3}
\begin{split}
C\Bigg( &\sum_{(q)\leq 2} \Big(\sqrt{E_2[\Omega^{(q)} Q^{(-)}]}+\sqrt{E_2[\Omega^{(q)} P]}\\& + \sqrt{E_2[\Omega^{(q)} \slashed{\nabla}_{\hat{Y}}P]}+ \sqrt{E_2[\Omega^{(q)} Q^{(+)}]}\Big) \Bigg)\tau^{-1/2}
\end{split}
\end{equation}
through the decay foliation \eqref{decayFoliation}.
\end{theorem}
\begin{proof}
We define
\begin{equation}
X^{\ell \geq 2}:= G + X,
\end{equation}
with $G$ and $X$ co-vectors from Lemma \ref{RWcovector} and \eqref{interpolatedCovector}, respectively imposing the Regge-Wheeler gauge on the co-closed portion and the interpolated Chandrasekhar gauge on the closed portion.  Combining Theorems \ref{coclosedMain} and \ref{closedMain}, the difference
\begin{equation}
\widehat{\delta g}^{\ell \geq 2} := \delta g^{\ell \geq 2} - \pi_{X^{\ell \geq 2}}
\end{equation}
has linearized metric components whose norms are controlled by the decaying quantity
\begin{align*}
C\Bigg( &\sum_{(q)\leq 2} \Big(\sqrt{E_2[\Omega^{(q)} Q^{(-)}]}+\sqrt{E_2[\Omega^{(q)} P]}\\& + \sqrt{E_2[\Omega^{(q)} \slashed{\nabla}_{\hat{Y}}P]}+ \sqrt{E_2[\Omega^{(q)} Q^{(+)}]}\Big) \Bigg)\tau^{-1/2}.
\end{align*}
\end{proof}

For more refined decay estimates of the co-closed and closed portions, and their associated linearized metric coefficients, we refer the reader to Theorems \ref{coclosedMain} and \ref{closedMain}.
\appendix

\section{Symmetric traceless two-tensors on $S^2$ }
In this appendix, several calculations regarding symmetric traceless two-tensors on the unit sphere are provided. The analysis carries over to such tensors defined on a spherically symmetric spacetime with respect to the operator $\mathring{\nabla}$ and the associated Laplacian $\mathring{\Delta}$ specified in Section 3.

Consider the standard unit sphere $S^2$ with round metric $\sigma_{\alpha\beta}$ and area form $\epsilon_{\alpha\beta}$, where $u^\alpha, \alpha=1, 2$ is a local coordinate system on $S^2$.  We use $\nabla$ to refer to the associated covariant derivative operator, and $\Delta$ to refer to the associated spherical Laplacian.

For one-forms and two-tensors on $S^2$, the following formula holds:
\begin{equation}\begin{split}(\nabla_{\alpha}\nabla_{\beta}-\nabla_{\beta}\nabla_{\alpha}) T_{\gamma}&=\delta_{\alpha\gamma}T_{\beta}-\delta_{\beta\gamma} T_{\alpha},\\
(\nabla_{\alpha}\nabla_{\beta}-\nabla_{\beta}\nabla_{\alpha}) T_{\gamma\eta}&=\delta_{\alpha\gamma}T_{\beta\eta}-\delta_{\beta\gamma} T_{\alpha\eta}+\delta_{\alpha\eta} T_{\gamma\beta}-\delta_{\beta\eta}T_{\gamma\alpha}.\end{split}\end{equation}

\begin{proposition}
Let $S_{\alpha\beta}$ be a symmetric traceless two-tensor on $S^2$.

(1) Then there exists a one-form $p_{\alpha}$ such that 
\begin{equation}\label{sym_2} S_{\alpha\beta}=\frac{1}{2}(\nabla_{\beta} p_{\alpha}+\nabla_{\alpha} p_{\beta}-(\nabla^{\gamma} p_{\gamma})\sigma_{\alpha\beta}).\end{equation}

(2) If $\nabla^{\alpha} S_{\alpha\beta}$ is closed, then $p_{\alpha}$ is closed and there exists a function $f$ such that 
\[S_{\alpha\beta}=\nabla_{\alpha}\nabla_{\beta} f-\frac{1}{2}\sigma_{\alpha\beta}\Delta f,\] and $\nabla^{\alpha} S_{\alpha\beta}=\frac{1}{2}\nabla_{\beta}(\Delta f+2f), \nabla^{\beta}\nabla^{\alpha} S_{\alpha\beta}=\frac{1}{2}\Delta(\Delta f+2f)$. 

(3) If $\nabla^{\alpha} S_{\alpha\beta}$ is co-closed, then there exists a function $g$ such that 
\[S_{\alpha\beta}=\frac{1}{2}(\epsilon_{\alpha}^{\gamma} \nabla_{\beta}\nabla_{\gamma} g+\epsilon_{\beta}^{\gamma} \nabla_{\alpha}\nabla_{\gamma} g),\] and $\epsilon^{\gamma\beta}\nabla_{\gamma}\nabla^{\alpha} S_{\alpha\beta}=\frac{1}{2} \Delta(\Delta g+2g)$.
\end{proposition}

\begin{proof}

Since there is no harmonic one-form on $S^2$, Hodge decomposition implies that any one-form $p_{\beta}$ can be written as  $p_{\beta}=\nabla_{\beta} f+\epsilon_{\beta}^{\gamma} \nabla_{\gamma} g$ for two functions $f$ and $g$ on $S^2$.  Similarly, a symmetric traceless two-tensor $S_{\alpha\beta}$ has a potential one-form $p_{\beta}$ as in the decomposition \eqref{sym_2}.  Taking a derivative of \eqref{sym_2}, we find
\[\frac{1}{2}(\nabla^{\alpha}\nabla_{\beta} p_{\alpha}+\nabla^{\alpha}\nabla_{\alpha} p_{\beta}-\nabla_{\beta} \nabla^{\alpha} p_{\alpha})=\nabla^{\alpha}S_{\alpha\beta}.\]

On $S^2$, we have the curvature relation \[\nabla^{\alpha}\nabla_{\beta} p_{\alpha}-\nabla_{\beta} \nabla^{\alpha} p_{\alpha}=p_{\beta}.\]

On the other hand, $p_{\beta}=\nabla_{\beta} f+\epsilon_{\beta}^{\gamma} \nabla_{\gamma} g$ implies
\[\nabla^{\alpha} \nabla_{\alpha}(p_{\beta})=\nabla_{\beta}(\Delta f)+\epsilon_{\beta}^{\gamma}\nabla_{\gamma} (\Delta g)+\nabla_{\beta} f+\epsilon_{\beta}^{\gamma} \nabla_{\gamma} g.\]

Altogether, we obtain
\[\frac{1}{2}\left[\nabla_{\beta}(\Delta f+2f)+\epsilon_{\beta}^{\gamma} \nabla_{\gamma}(\Delta g+2g)\right]=\nabla^{\alpha}S_{\alpha\beta}.\]

Examining the expression above, we see that the first parts of (2) and (3) follow.  It remains to show
\begin{equation}\begin{split}\label{f_g}\frac{1}{2} \Delta(\Delta f+2f)&=\nabla^{\beta}\nabla^{\alpha} S_{\alpha\beta},\\
\frac{1}{2} \Delta(\Delta g+2g)&=\epsilon^{\gamma\beta}\nabla_{\gamma}\nabla^{\alpha} S_{\alpha\beta}.\end{split}\end{equation}
We apply the operator $\nabla^{\beta}$ and $\epsilon^{\eta\beta}\nabla_{\eta}$ to both sides and obtain the desired equation \eqref{f_g}.
Let $X^i, i=1, 2, 3$ be the three coordinate functions for the standard embedding of $S^2$ into $\R^3$. It is known that they form a basis of the eigenspace of eigenvalue $-2$. The kernel of the operator $\Delta +2$ on $S^2$ consists of exactly this eigenspace.
In order to solve for $f$ and $g$, we check that:
\[\int_{S^2} (\nabla^{\beta}\nabla^{\alpha} S_{\alpha\beta})X^i=-\int \sigma^{\alpha\beta} S_{\alpha\beta} X^i=0,\] and
\[\int_{S^2}(\epsilon^{\gamma\beta}\nabla_{\gamma} \nabla^{\alpha} S_{\alpha\beta})X^i=\int \epsilon^{\alpha\beta} S_{\alpha\beta} X^i=0,\] because $S_{\alpha\beta}$ is symmetric and trace-free.  Therefore both $f$ and $g$ can be solved and the ambiguity consists exactly of the first eigenfunctions $X^i$ and constant functions. 
\end{proof}

\begin{proposition}\label{oneformHarmonics}
If $f$ is an eigenfunction on $S^2$ with $\Delta f=-\lambda f$, then both $\nabla_{\alpha}f$ and $\epsilon_{\alpha}^{\gamma} \nabla_{\gamma} f$ are eigensections, with $\nabla^{\alpha}\nabla_{\alpha} \nabla_{\beta} f=(-\lambda+1)\nabla_{\beta} f$ and similarly. 
\end{proposition}
\begin{proof}
We use the formula \[\nabla^{\alpha}\nabla_{\beta} p_{\alpha}-\nabla_{\beta} \nabla^{\alpha} p_{\alpha}=p_{\beta}\] for any one-form.
\end{proof}
\begin{proposition}\label{twoformHarmonics}
If $p_{\alpha}$ is an eigensection on $S^2$ with $\nabla^{\gamma} \nabla_{\gamma} p_{\alpha}=(-\lambda+1) p_{\alpha}$, then $S_{\alpha\beta}=\frac{1}{2}(\nabla_{\beta} p_{\alpha}+\nabla_{\alpha} p_{\beta}-(\nabla^{\gamma} p_{\gamma})\sigma_{\alpha\beta})$ is an eigensection and 
\[\nabla^{\gamma} \nabla_{\gamma} S_{\alpha\beta}=(-\lambda+4) S_{\alpha\beta}.\]
\end{proposition}

\bibliographystyle{plain}
\bibliography{schwarzStabilityBibModified}

\end{document}